\pdfoutput=1                                              
\documentclass[12pt]{article}
\usepackage[T1]{fontenc}
 \usepackage{stix}
 \usepackage{libertine,libertinust1math}
\usepackage[dvipsnames]{xcolor}
\usepackage{pdfpages}
\usepackage{amsmath}
\usepackage{amsfonts}
\usepackage{sgame}
\usepackage{accents}
\usepackage{tikz-cd}
\usepackage{float}
\usepackage[numbers]{natbib}   
\usepackage{multirow}
\usepackage{dutchcal}
\usepackage{pdfpages}
\usepackage{sgame}
\usepackage{mathtools}
\usepackage{amsthm}
\usepackage{enumerate}
\usepackage{epigraph}
\usepackage{sectsty}
\usetikzlibrary{calc}
\usetikzlibrary{arrows}

\newtheorem{theorem}{Theorem}[section]

\newtheorem{lemma}[theorem]{Lemma}
\newtheorem{corollary}[theorem]{Corollary}
\newtheorem{claim}[theorem]{Claim}

\theoremstyle{definition}

\newtheorem{definition}[theorem]{Definition}

\newtheorem{condition}[theorem]{Condition}

\usepackage{hyperref}
\hypersetup{
    colorlinks=true,
    linkcolor=ForestGreen,
    filecolor=magenta,      
    urlcolor=Plum,
    citecolor=ForestGreen,
}

\providecommand{\keywords}[1]{\textbf{\textit{Keywords:}} #1}
\providecommand{\jel}[1]{\textbf{\textit{JEL Classifications:}} #1}
\begin{document}
\title{In Simple Communication Games, When Does \textit{Ex Ante} Fact-Finding Benefit the Receiver?}
\author{Mark Whitmeyer\thanks{Department of Economics, University of Texas at Austin \newline Email: \href{mailto:mark.whitmeyer@utexas.edu}{mark.whitmeyer@utexas.edu}. Thanks to Rosemary Hopcroft, Vasudha Jain, Meg Meyer, Alex Teytelboym, Thomas Wiseman, and Joseph Whitmeyer for their comments and suggestions. This paper contains material that had previously been in the job market paper, ``Bayesian Elicitation''.}}
\date{\today{}}
\maketitle
\begin{abstract}
Always, if the number of states is equal to two; or if the number of receiver actions is equal to two and \begin{enumerate}[i.]
    \item The number of states is three or fewer, or 
    \item The game is cheap talk, or
    \item There are just two available messages for the sender.
\end{enumerate}
A counterexample is provided for each failure of these conditions.
\end{abstract}
\\
\jel{C72; D82; D83}

\newpage

\section{Introduction}\label{intro}
\setlength{\epigraphwidth}{2.3in}
\epigraph{A little learning is a dangerous thing.}{Alexander Pope\\ \textit{An Essay on Criticism}}

This is a paper on the value of information. The setting is a two player sender-receiver, signaling, or communication, game. There is an unknown state of the world, about which the receiver is uninformed. The receiver is faced with a decision problem but has no direct access to information about the state. Instead, there is an informed sender, who, after learning the state, chooses a (possibly costly) action (message)\footnote{Throughout, in order to distinguish the sender's action from the receiver's action, the sender's action is termed a message. In some settings, like cheap talk games, this moniker is literal. In some settings, like e.g. the classic Spence scenario, in which the sender chooses a level of education, message is less fitting as a label. Hence, the reader should keep in mind that the message is simply the receiver's action.}, which the receiver observes. 

We make a minimal number of assumptions. Each player has a von Neumann-Morgenstern utility function that may depend on the message chosen by the sender, the state of the world, and the action chosen by the receiver. Thus, these games include cheap talk games as in Crawford and Sobel (1982) \cite{cs}, and signaling games as in Spence (1978) \cite{spence} or Cho and Kreps (1986) \cite{cho}. Throughout we assume that the number of states, messages, and actions are finite.

In this setting, we pose a simple question. Is the receiver's maximal equilibrium payoff convex in the prior? That is, restricting attention to the equilibrium that maximizes the receiver's expected payoff, does \textit{ex ante} learning always benefit the receiver? If not, then are there conditions that guarantee this convexity?

We show that the answer to the first question is no: the receiver's maximal equilibrium payoff is not generally convex in the prior. However, there are broad conditions that guarantee convexity. If the game is \hyperlink{simple}{simple}--the sender's message has only instrumental value to the receiver--then the receiver's payoff is convex in the prior provided either
\begin{enumerate}
    \item There are at most two states; or
    \item The receiver has at most two actions and 
    \begin{enumerate}[i.]
        \item The game is cheap talk; or
        \item There are at most two messages; or
        \item There are at most three states.
    \end{enumerate}
\end{enumerate}
If there are three or more messages, four or more states, and the game is not cheap talk, then even if the game is simple and the receiver has just two actions, the receiver's payoff may fail to be convex in the prior. Moreover, if there are three or more states and the receiver has three or more actions, then the receiver's payoff may fail to be convex in the prior, even if the game is simple and cheap talk with transparent motives (a cheap talk game in which each sender has identical preferences over the action chosen by the receiver). Furthermore, if the game is non-simple then the receiver's payoff may fail to be convex in the prior, even if there are just two states and two actions. 

Why is the receiver's payoff convex in those scenarios described above? Why may the payoff fail to be convex otherwise? There is a crucial trade-off that belongs to \textit{ex ante} information acquisition: there is an initial gain in information that, all else equal, benefits the receiver. However, all else may not be equal: the initial learning may result in a belief at which the receiver-optimal equilibrium may be quite bad for the receiver. Hence, the two effects may have opposite effects on the receiver's welfare, in which case the magnitude of each effect determines whether learning is beneficial.

The conditions described above guarantee that the first effect dominates--even if the resulting beliefs after learning lead to worse equilibria for the receiver, her welfare loss is guaranteed to be less than the welfare gain from the information acquisition itself. If the conditions do not hold, then the first effect may not dominate. Even though the receiver gains information initially, the resulting equilibria may be so bad that the receiver may strictly prefer not to learn.

Thanks to the ubiquity of communication games, there are numerous interpretations of \textit{ex ante} information acquisition. In the Spence (1978) \cite{spence} setting, this paper's question becomes, ``when does any test (prior to the sender's education choice) benefit the hiring firm(s)?'' A seminal paper in finance is Leland and Pyle (1977) \cite{lp}, who explore an entrepreneur signaling through his equity retainment decision. There, ``when does any background information or access to the entrepreneur's history benefit a prospective investor?'' In a political economy setting in which an incumbent signals through his policy choice (see e.g. Angeletos, Hellwig, and Pavan 2006, and Caselli, Cunningham, Morelli, and de Barreda 2014) \cite{pavan, ines}, we ask, ``when does any initial news article benefit a representative member of the populace?''

More applications of \textit{ex ante} information acquisition include reports about the state of the economy, in the case of a central bank signaling through its monetary policy (Melosi 2016) \cite{mel}; product reviews, in the case of a firm signaling through advertising (Nelson 1974, and Milgrom and Roberts 1986) \cite{nelson, mil}, or through its warranty offer (Gal-Or 1989) \cite{gal}; and financial reports or audits, in the case of a firm signaling through dividend provision (Bhattacharyya 1980) \cite{bat}.

The remainder of Section \ref{intro} discusses related work, and Section \ref{model} describes the formal model. Sections \ref{staypositive} and \ref{negative} contain the main results of the paper, Theorems \ref{main1} and \ref{main2}, which provide sufficient conditions for convexity and show that the receiver's payoff may not be convex should those conditions not hold, respectively. Section \ref{conclusion} concludes.

\subsection{Related Work}

One way to rephrase this paper's research question is, ``if information is free prior to a communication game, then does it benefit the receiver in expectation to acquire it?" Ramsey (1990) \cite{ram} asks this question in the context of a decision problem and answers in the affirmative, and this result also follows from Blackwell (1951, 1953) \cite{blackwell, blackwell2} among many others.

There are a number of papers that investigate the value of information in strategic interactions (games). Neyman (1991) \cite{ney} shows that information can only help a player in a game if other players are unaware that she has it. Kamien, Tauman, and Zamir (1990) \cite{kamien} explore an environment in which an outside agent, ``the Maven'', possesses information relevant to an $n$-player game in which he is not a participant. There they look at the outcomes that the maven can induce in the game and how (and for how much) the maven should sell the information. Bassan, Gossner, Scarsini, and Zamir (2003) \cite{bass} establish necessary and sufficient conditions for the value of information to be socially positive in a class of games with incomplete information. 

In two-player (simultaneous-move) Bayesian games, Lehrer, Rosenberg, and Shmaya (2013) \cite{sh2} forward a notion of equivalence of information structures as those that induce the same distributions over outcomes. They characterize this equivalence for several solution concepts, including Nash equilibrium. In a companion paper, they (Lehrer, Rosenberg, and Shmaya 2010) \cite{sh} look at the same set of solution concepts in (two-player) common interest games and characterize which information structures lead to higher (maximal) equilibrium payoffs. Gossner (2000) \cite{GOSS1} compares information structures through their ability to induce correlated equilibrium distributions, and Gossner (2010) \cite{GOSS2} introduces a relationship between ``ability'' and knowledge: not only does more information imply a broader strategy set, but a converse result holds as well. 

Ui and Yoshizawa (2015) \cite{ui} explore the value of information in (symmetric) linear-quadratic-Gaussian games and provide necessary and sufficient conditions for (public or private) information to increase welfare. Kloosterman (2015) \cite{close} explores (dynamic) Markov games and provides sufficient conditions for the set of strongly symmetric subgame perfect equilibrium payoffs of a Markov game to decrease in size (for any discount factor) as the informativeness of a public signal about the next period's game increases. Gossner and Mertens (2001) \cite{goos}, Lehrer and Rosenberg (2006) \cite{teacher}, P\k{e}ski (2008) \cite{pes}, and De Meyer, Lehrer, and Rosenberg (2010) \cite{meyer} all study the value of information in zero-sum games.

In a sense, this paper explores the decision problem faced by the receiver in which the information she obtains is endogenously generated by equilibrium play by the sender. That is, the receiver's problem is one in which \textit{ex ante} information acquisition results in a (possibly) different information generation process at the resulting posterior belief. Outside of that there are no strategic concerns; and the sender is perfectly informed, so there is no learning on his part. Consequently, this paper is more similar in spirit to the original question asked by Ramsey, and we need not concern ourselves with the possible complexity of information structures for multiplayer games of incomplete information.

Furthermore, this paper investigates the value of information in communication games, which are by definition games of information transmission. In contrast to the broad class of games of incomplete information, in communication games the transfer of information between sender and receiver is of paramount importance. The main results of this paper pertain to a restriction of that class of games--\hyperlink{simple}{simple} games--in which the sender's message affects the receiver's payoff only through the information that it contains.

The paper closest to this one is its companion paper, Whitmeyer (2019) \cite{Whit}, which investigates how a receiver can design an information structure in order to optimally elicit information from a sender in a communication game. There, in a two player communication game, the receiver may commit \textit{ex ante} to a signal $\pi: M \to \Delta(X)$, where $X$ is a (compact) set of signal realizations. Instead of observing the sender's message, the receiver observes a signal realization correlated with the message. In one of the main results of that paper, we discover that in simple two-action games this ability guarantees that the value of information is always positive. Contrast this to the negative result that we find in this paper--that in simple two-action games the value of information is not generally positive--the other paper turns this on its head and shows that information design guarantees a positive value of information.

\section{The Model}\label{model}

There are two players: an informed sender, $S$; and a receiver, $R$, who share a common prior about the state of the world, $\mu_{0} \in \Delta(\Theta)$, where $\mu_{0}(\theta) = \Pr(\Theta = \theta)$. There are two stages to the scenario--first, there is a learning stage. 

\textbf{Stage 1 (Learning Stage):}
There is some finite (or at least compact) set of signal realizations $Y$ and a signal or Blackwell experiment, mapping $\zeta: \Theta \to \Delta(Y)$ whose realization is public. This experiment leads to a distribution over posteriors, where the posterior following signal realization $y$ is $\mu_{y}$. Call $\zeta$ the \hypertarget{ie}{\textcolor{Plum}{Initial Experiment}}. Each signal realization begets (via Bayes' law) a posterior distribution, $\mu_{y}$. Thus, experiment $\zeta$ leads to a distribution over posterior distributions, $P \in \Delta \Delta \left(\Theta\right)$, whose average is the prior distribution:
    \[\mathbb{E}_{P}\left[\mu\right] \equiv \int_{\Delta\left(\Theta\right)}\mu dP(\mu) = \mu_{0}\]

Each posterior is the prior for the ensuing communication game. That is, following each realization of the experiment, the sender and receiver then take part in a second stage, the communication game.

\textbf{Stage 2 (Communication Game):} In this stage, $S$ and $R$ share the common prior $\mu_{y}$. The sender has private information, his type (or the state of the world), $\theta \in \Theta$: he observes his type before choosing a message, $m$, from a set of messages $M$. The receiver observes $m$, but not $\theta$, updates her belief about the receiver's type and message using Bayes' law, then chooses a mixture over actions, $A$. We assume that these sets, $M, A$ and $\Theta$, are finite.

Each player, $S$ and $R$, has preferences over the message sent, the action taken, and the type of the sender. These are represented by the utility functions\footnote{Since the domain is finite, any $u_{i}$ is continuous.} $u_{i}$, $i \in \left\{S,R\right\}$: $u_{i}: M \times A \times \Theta \to \Re$.

Let us revisit the timing. First, there is an initial experiment which begets a distribution over (common) posterior beliefs, which are each respectively (common) prior beliefs in the ensuing communication game. Second, $S$ observes his private type $\theta \in \Theta$, and chooses a message $m \in M$ to send to $R$. $R$ observes $m$, updates his belief, and chooses action $a \in A$.

We extend the utility functions for the players to behavioral strategies. A behavioral strategy for $S$, $\sigma_{\theta}(m)$ is a probability distribution over $M$; it is the probability that a type $\theta$ sender sends message $m$. Similarly, a behavioral strategy for $R$, $\rho(a | m)$ is a probability distribution over $A$; it is the probability that the receiver chooses action $a$ following message $m$.

We focus on receiver-optimal Perfect Bayesian Equilibrium (PBE), which we define in the standard manner. Henceforth by equilibrium or PBE, we refer to those particular equilibria, and by receiver's payoff we mean the receiver's payoff in the receiver-optimal PBE.

Throughout, we consider various sub-classes of communication games. These sub-classes are defined as follows
\begin{definition}
A communication game is \hypertarget{simple}{\textcolor{Plum}{Simple}} if the receiver has preferences over the action taken and the type of the sender, but not over the message chosen by the sender. Equivalently, a game is simple provided the receiver's preferences are represented by the utility function $u_{R}: A \times \Theta \to \Re$.
\end{definition}
On occasion, we derive results that hold for two other classes of communication games; cheap talk, and cheap talk with transparent motives. We remind ourselves of their definitions:
\begin{definition}
A communication game is \hypertarget{ct}{\textcolor{Plum}{Cheap Talk}} if the sender has preferences over the action taken by the receiver and the type of the sender, but not over the message he chooses. Namely, for each type, each message is equally costless. Equivalently, a game is cheap talk provided the sender's preferences are represented by utility function $u_{S}: A \times \Theta \to \Re$.
\end{definition}
A subclass of the class of cheap talk games are those with transparent motives, which term was introduced in Lipnowski and Ravid (2017) \cite{transparent}:
\begin{definition}
A communication game is \hypertarget{ctwt}{\textcolor{Plum}{Cheap Talk with Transparent Motives}} if the game is cheap talk and the sender's preferences over the action taken by the receiver are independent of his type. Equivalently, a game is cheap talk with transparent motives provided the sender's preferences are represented by the utility function $u_{S}: A \to \Re$.
\end{definition}

\section{When the Value of Information is Always Positive}\label{staypositive}

This section is devoted to establishing the following theorem, which provides sufficient conditions for the value of information to always be positive in communication games.

\begin{theorem}\label{main1}
In simple communication games, the value of information is always positive for the receiver provided
\begin{enumerate}
    \item There are two states (or fewer); or 
    \item The receiver has two actions (or fewer) and
\begin{enumerate}[i.]
    \item There are three states (or fewer); or
    \item There are two messages (or fewer); or
    \item The game is cheap talk.
\end{enumerate}
\end{enumerate}
\end{theorem}

To begin, we show that if there are two states of the world (or two types of sender), the receiver's payoff is convex in the prior. Observe that if there is no initial experiment, and the sender and receiver participate in the signaling game with common prior $\mu_{0}$, then there exists a signal or experiment $\eta: \Theta \to \Delta(M)$ that is induced by the optimal equilibrium. This experiment leads to a distribution over posteriors, where the posterior following message $m$ is $\mu_{m}$. Call this experiment the \hypertarget{ne}{\textcolor{Plum}{Null-Optimal Experiment}}.

\begin{lemma}\label{prop22}
In any \textit{simple} communication game with two states and $n$ actions, the receiver's payoff is convex in the prior. 
\end{lemma}

\begin{proof}

We sketch the proof here and leave the details to Appendix \ref{22proof}. The first step is to establish Claim \ref{steak}, which allows us to restrict the number of messages in the game to two without loss of generality. As a result, there are just three cases that we need to consider: first, where the sender types pool in the receiver-optimal equilibrium at belief $\mu_{0}$; second, where one sender type mixes and the other chooses a pure strategy (in the receiver-optimal equilibrium at belief $\mu_{0}$); and third, where both sender types mix. Note that this lemma holds trivially if there exists a separating equilibrium, so we need not consider that case.

Next, following any realization of the initial experiment, $y$, there exists a receiver-optimal equilibrium. Equivalently, there exists a signal or experiment $\gamma_{y}: \Theta \to \Delta(M)$ that is induced by the optimal equilibrium. This experiment leads to a distribution over posteriors, where the posterior following message $m$ is $\mu_{m}$. Call this experiment the \hypertarget{se}{\textcolor{Plum}{y-equilibrium experiment}}.

Then, we define $\xi$ as the experiment that corresponds to the information ultimately acquired by the receiver following the initial learning and the resulting equilibrium play in the signaling game. All that remains is to show in each of the three cases that the null-optimal experiment, $\eta$, is less Blackwell informative than $\xi$ and so the receiver prefers $\xi$--the receiver prefers any learning.

Note that it is possible to ``prove this result without words", which proof is depicted in Figure \ref{case1}. In each case, the red point corresponds to the prior, the blue arrows and points to the initial experiment and posteriors, the yellow arrows and points to the null-optimal experiment, the green arrows to the y-equilibrium experiments, and the purple arrows and points to experiment $\xi$.

\begin{figure}
\begin{center}
\includegraphics[scale=.33]{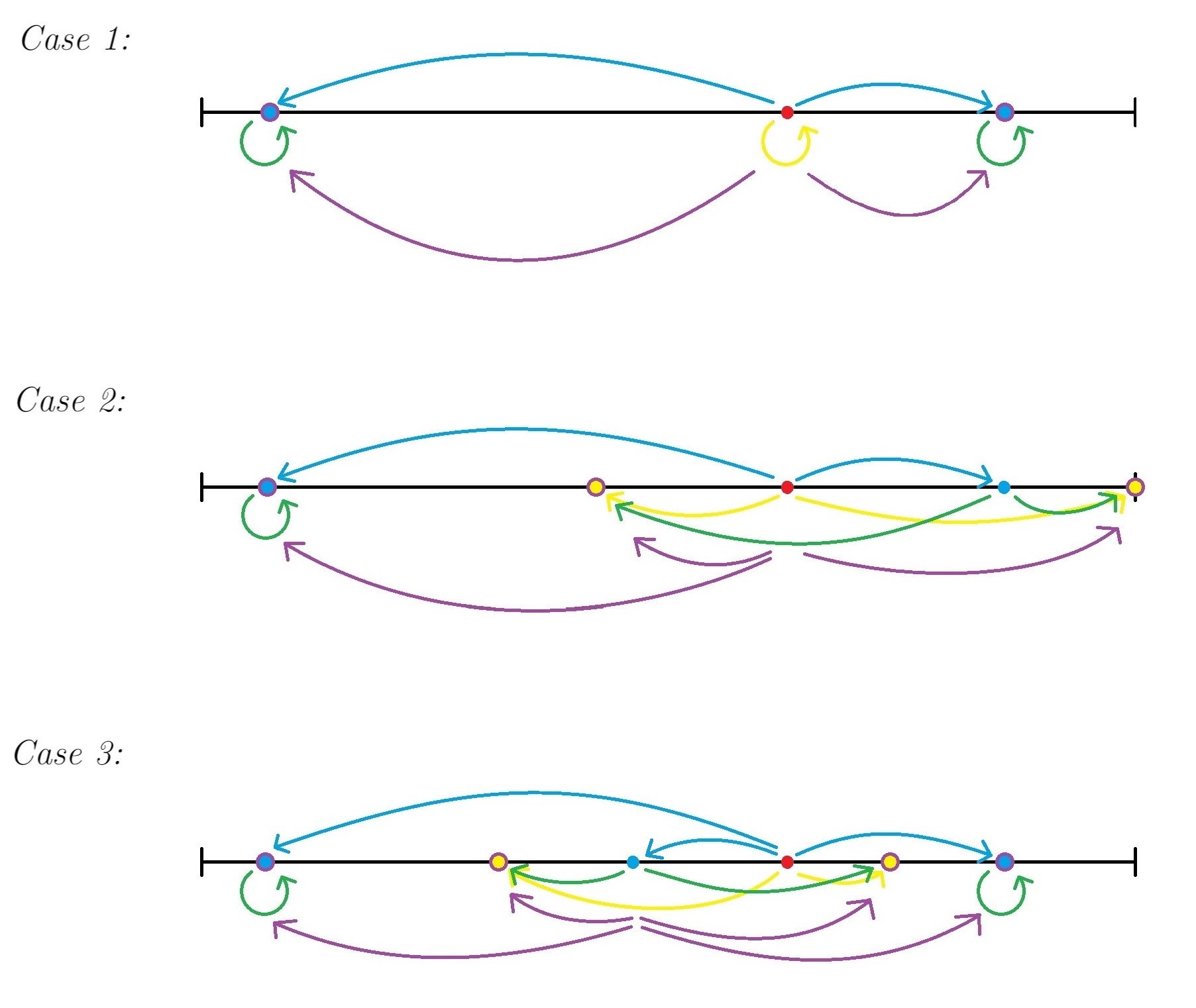}
\end{center}
\caption{\label{case1} Lemma \ref{prop22} Proof}\end{figure}

\end{proof}

Next, we explore convexity when the receiver has only two actions. First, we establish that it is without loss of generality to restrict attention to equilibria in which no type mixes over messages at which the receiver strictly prefers different actions.

\begin{lemma}\label{nodivide}
In simple games, there exists a receiver-optimal equilibrium in which no type of sender mixes over messages that induce beliefs at which the receiver strictly prefers different actions.
\end{lemma}
\begin{proof}
The proof is left to Appendix \ref{ndproof}
\end{proof}

Second, we discover that if there is a receiver optimal equilibrium at belief $\mu_{0}$ in which at most two messages are used, then any information benefits the receiver. Formally,

\begin{lemma}\label{too}
Consider any simple communication game. If there is a receiver-optimal equilibrium at belief $\mu_{0}$ in which at most two messages are used, then any initial experiment benefits the receiver. 
\end{lemma}
\begin{proof}
The full proof is left to Appendix \ref{tooproof}.
\end{proof}

Because of the costless nature of messages in cheap talk games, in conjunction with Lemma \ref{nodivide}, it is clear that there must be a receiver-optimal equilibrium at belief $\mu_{0}$ in which at most two messages are used. Accordingly, Lemma \ref{too} implies

\begin{corollary}
In any $n$ state, two action, simple cheap talk game, the receiver's payoff is convex in the prior.
\end{corollary}

From Lemma \ref{prop22} we know that in two state, two action simple communication games, the value of information is always positive for the receiver. Perhaps surprisingly, the value of information is also always positive for the receiver in three state, two action simple communication games. \textit{Viz},

\begin{lemma}\label{409}
In simple communication games, for three states and two actions, the receiver's payoff is convex in the prior.
\end{lemma}

\begin{proof}

Again, we leave the detailed proof to Appendix \ref{409proof} but provide a sketch here. From Lemma \ref{nodivide}, we conclude that there is a receiver-optimal equilibrium at $\mu_{0}$ in which at most three messages are used. If two messages or fewer are used, then from Lemma \ref{too}, we have convexity. Thus, it remains to consider the case in which three messages are used. Fortunately, we show that there is just one such equilibrium that we need to consider.

\begin{figure}
    \centering
    \includegraphics[scale=.28]{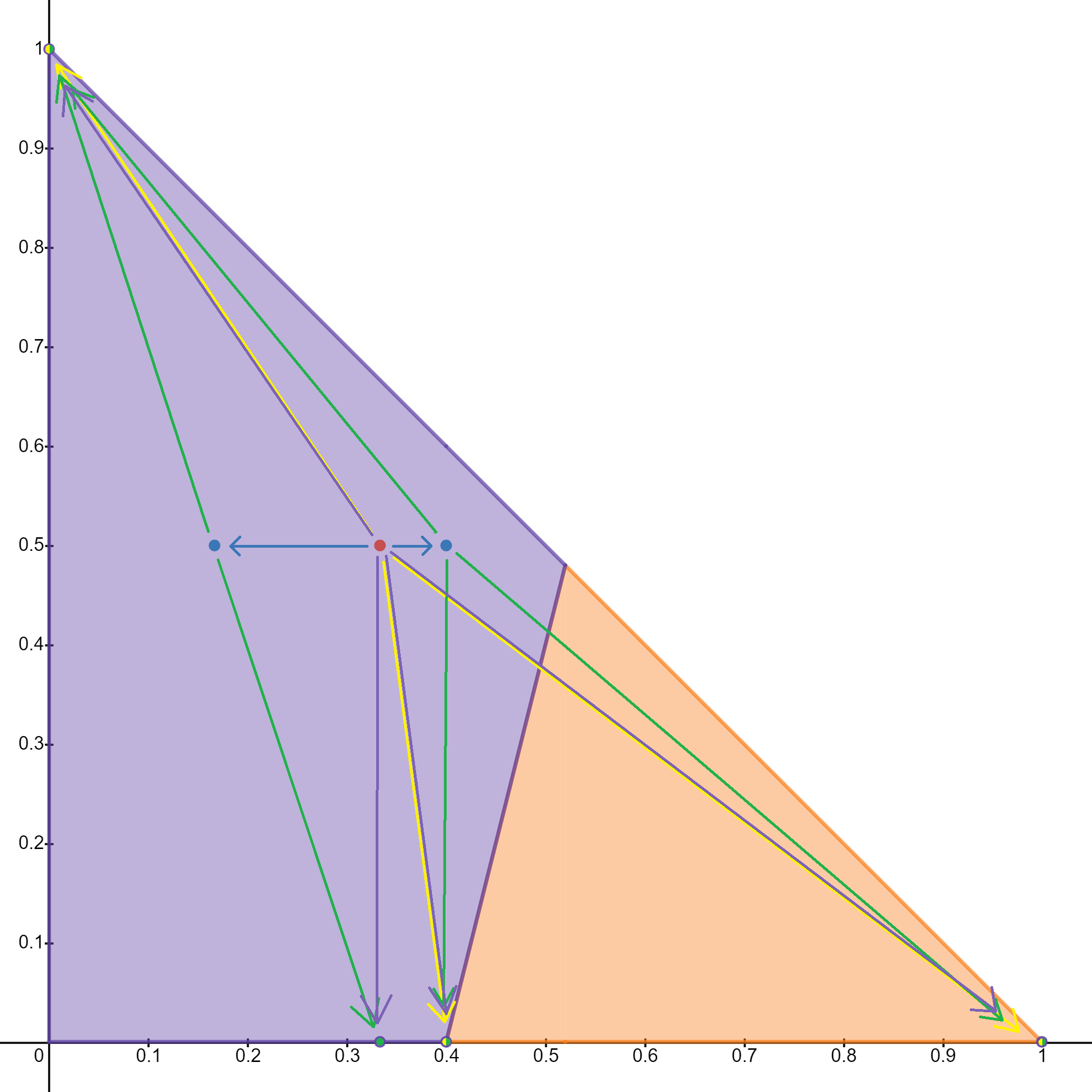}
    \caption{Lemma \ref{409} Proof}
    \label{3ow}
\end{figure}

Like Lemma \ref{prop22}, it is also possible to prove Lemma \ref{409} without words, which proof is depicted in Figure \ref{3ow}. The red point corresponds to the prior, the blue arrows and points to the initial experiment and posteriors, the yellow arrows and points to the null-optimal experiment, the green arrows to the y-equilibrium experiments (or rather experiments that are payoff-equivalent to the y-equilibrium experiments), and the purple arrows and points to experiment $\xi$.

\end{proof}

\section{When the Value of Information is not Always Positive}\label{negative}

This section tempers the optimism inspired by the Section \ref{staypositive}. Namely, we establish Theorem \ref{main2}, which states that if none of the sufficient conditions from Theorem \ref{main1} hold in some communication game, then there may be information that hurts the receiver

\begin{theorem}\label{main2}
In the following communication games, \textit{ex ante} information may hurt the receiver:
\begin{enumerate}
    \item Simple games with four or more states, three or more messages, and two actions;
    \item Simple games with three or more states and actions, and two or more messages;
    \item Non-simple games with two or more states, actions, and messages.
\end{enumerate}
\end{theorem}

We begin by proving Lemma \ref{410}, the first result listed in Theorem \ref{main2}.

\begin{figure}
\begin{center}
\includegraphics[scale=.7]{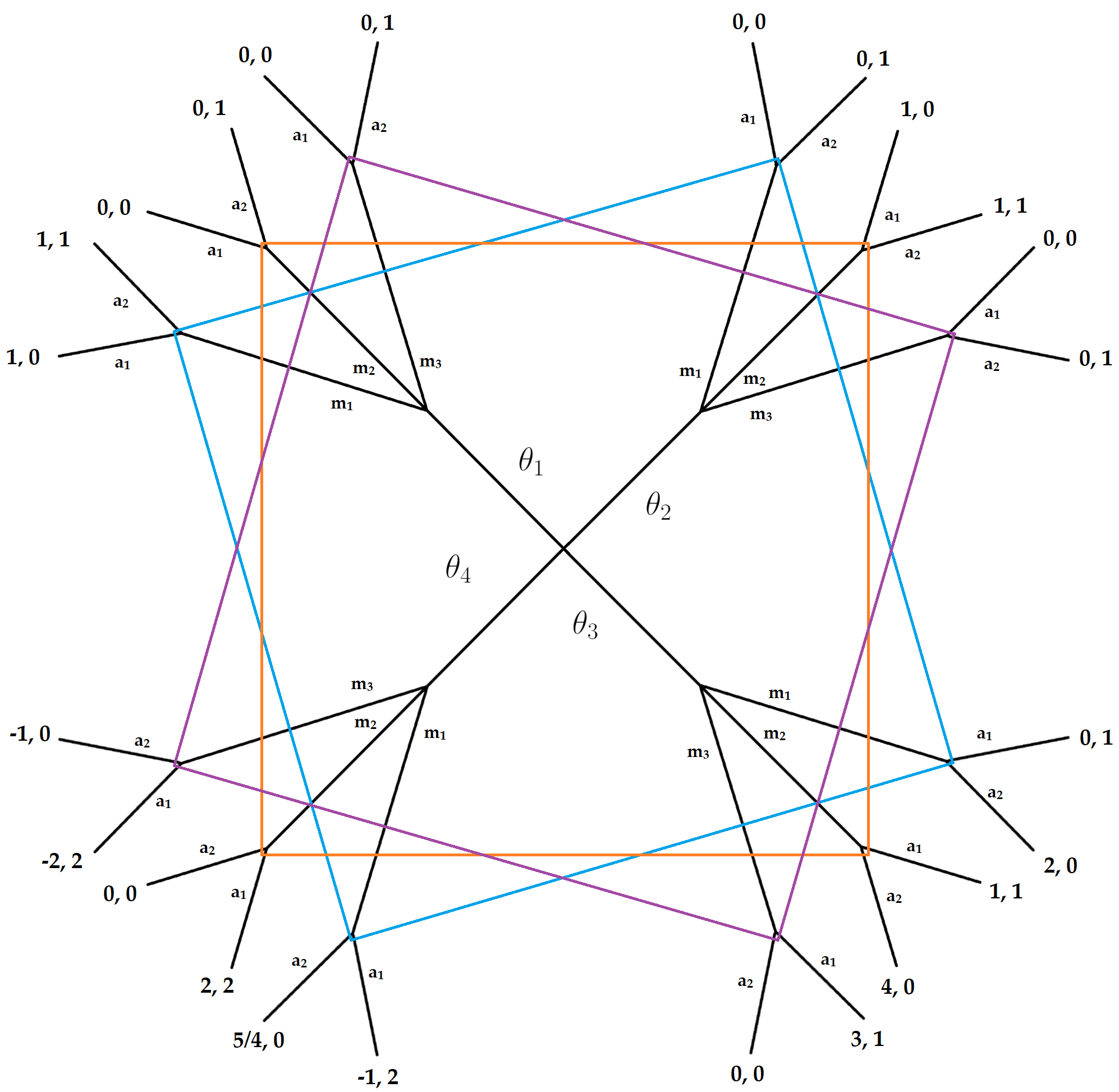}
\end{center}
\caption{\label{4c} Lemma \ref{410} Game}\end{figure}

\begin{lemma}\label{410}
In simple communication games, for four or more states, three or more messages, and two actions, the receiver's payoff is not generally convex in the prior.
\end{lemma}

\begin{proof}
Proof is via counter-example. There are four states, $\Theta = \left\{\theta_{1}, \theta_{2}, \theta_{3}, \theta_{4}\right\}$, and a belief is a quadruple $(\mu_{1}, \mu_{2}, \mu_{3},\mu_{4})$, where $\mu_{i} \coloneqq \Pr\left(\Theta=\theta_{i}\right)$ for all $i = 1, 2, 3, 4$ and $\mu_{1} + \mu_{2} + \mu_{3} + \mu_{4} = 1$.

The belief can be fully described with just three variables; hence, depicting the receiver's payoff as a function of the belief requires four dimensions. This is (rather) difficult to do, so instead we will restrict attention to a family of experiments that involve learning on just one dimension. That is, we fix $\mu_{1} = 1/3$ and $\mu_{3} = 1/8$, and consider only the receiver's payoff as a function of her (prior) belief about states $\theta_{2}$ and $\theta_{4}$. Learning is on just one dimension, and so (abusing notation) we rewrite the receiver's belief $\mu_{2}$ as $\mu$ and $\mu_{4}$ as $13/24 - \mu$, where $\mu \in [0,13/24]$.

In states $\theta_{1}$ and $\theta_{2}$, action $a_{2}$ is the correct action for the receiver; and in states $\theta_{3}$ and $\theta_{4}$, action $a_{1}$ is correct:

\begin{center}
\begin{tabular}{ |c|c|c|c|c| } 
\hline
Action & $\theta_{1}$ & $\theta_{2}$ & $\theta_{3}$ & $\theta_{4}$ \\ 
$a_{1}$ & $0$ & $0$ & $1$ & $2$ \\ 
$a_{2}$  & $1$ & $1$ & $0$ & $0$ \\ 
\hline
\end{tabular}
\end{center}

\medskip

Likewise, the sender's state (type)-dependent payoffs from message, action pairs are given as follows:

\medskip

\begin{center}
\begin{tabular}{ |c|c|c|c|c| } 
\hline
type & $\theta_{1}$ & $\theta_{2}$ & $\theta_{3}$ & $\theta_{4}$ \\
\hline
message & $m_{1}$ \quad $m_{2}$ \quad $m_{3}$ & $m_{1}$ \quad $m_{2}$ \quad $m_{3}$ & $m_{1}$ \quad $m_{2}$ \quad $m_{3}$ & $m_{1}$ \quad $m_{2}$ \quad $m_{3}$\\
\hline
$a_{1}$ & $1$ \quad $0$ \quad $0$ & $0$ \quad $1$ \quad $0$ & $0$ \quad $1$ \quad $3$ & $-1$ \quad $2$ \quad $-2$\\
\hline
$a_{2}$ & $1$ \quad $0$ \quad $0$ & $0$ \quad $1$ \quad $0$ & $2$ \quad $4$ \quad $0$ & $5/4$ \quad $0$ \quad $-1$\\
\hline
\end{tabular}
\end{center}

\medskip

Note that types $\theta_{1}$ and $\theta_{2}$ have messages that are \textit{strictly dominant} ($m_{1}$ and $m_{2}$, respectively), and that $\theta_{4}$ has a message that is \textit{strictly dominated} ($m_{3}$). 

\begin{figure}
    \centering
    \includegraphics[scale=.7]{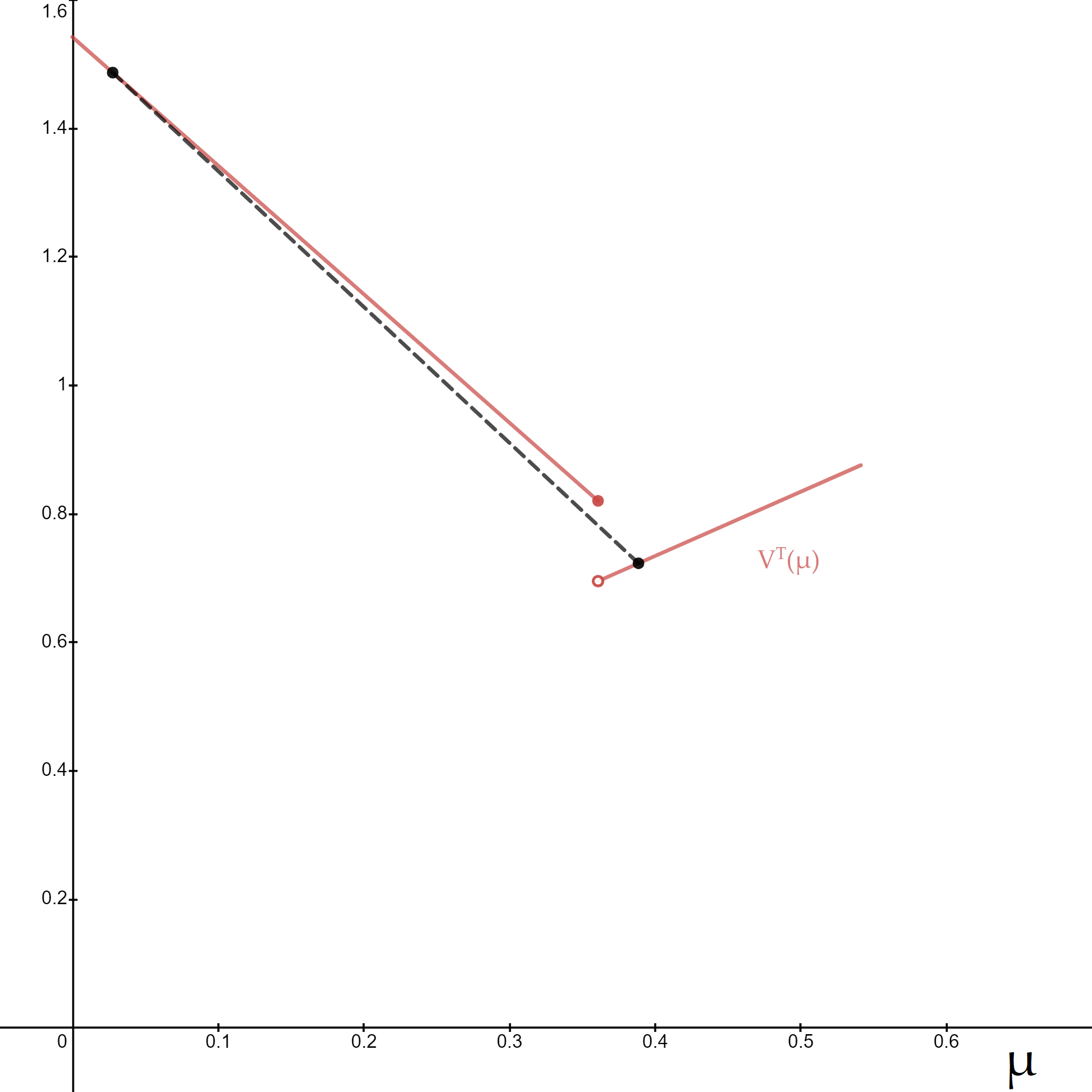}
    \caption{Receiver Payoffs (Lemma \ref{410} Proof)}
    \label{4dconvex}
\end{figure}

Figure \ref{4dconvex} depicts the receiver's equilibrium payoff as a function of $\mu$, $V^{T}$.\footnote{The super-script $T$ refers to ``transparency.'' This notation is due to the fact that in Whitmeyer (2019) \cite{Whit} we explore the value of information in the case when the receiver can choose the information structure in the ensuing game. There, we contrast the value of information in that ``optimal transparency" setting to the setting with full transparency, the focus of this paper.} Explicitly, that function is
\[V^{T}(\mu) = \begin{cases}
\frac{37}{24}- 2\mu & \mu \leq \frac{13}{36}\\
\frac{1}{3} + \mu & \frac{13}{36} < \mu \leq \frac{13}{24}
\end{cases}\]
and its derivation is left to Appendix \ref{410proof}. The receiver's payoff is no longer convex in the belief--in fact, it is no longer upper-semicontinuous. If $\mu > 13/36$--the receiver becomes too sure that the sender is not type $\theta_{3}$ or $\theta_{4}$--the only equilibria beget the pooling payoff, which correspond to no (or at least no useful) information transmission. The dotted line is a secant line that corresponds to a binary initial experiment that strictly hurts the receiver.
\end{proof}

With three or more states and actions, information may harm the receiver:

\begin{lemma}\label{3by3}
If there are at least three states, three actions and two messages then the receiver's payoff is not generally convex in the prior.
\end{lemma}
\begin{proof}
Proof is via counterexample. Consider the game depicted in Figure \ref{convexfig}. \begin{figure}
\begin{center}
\includegraphics[scale=.7]{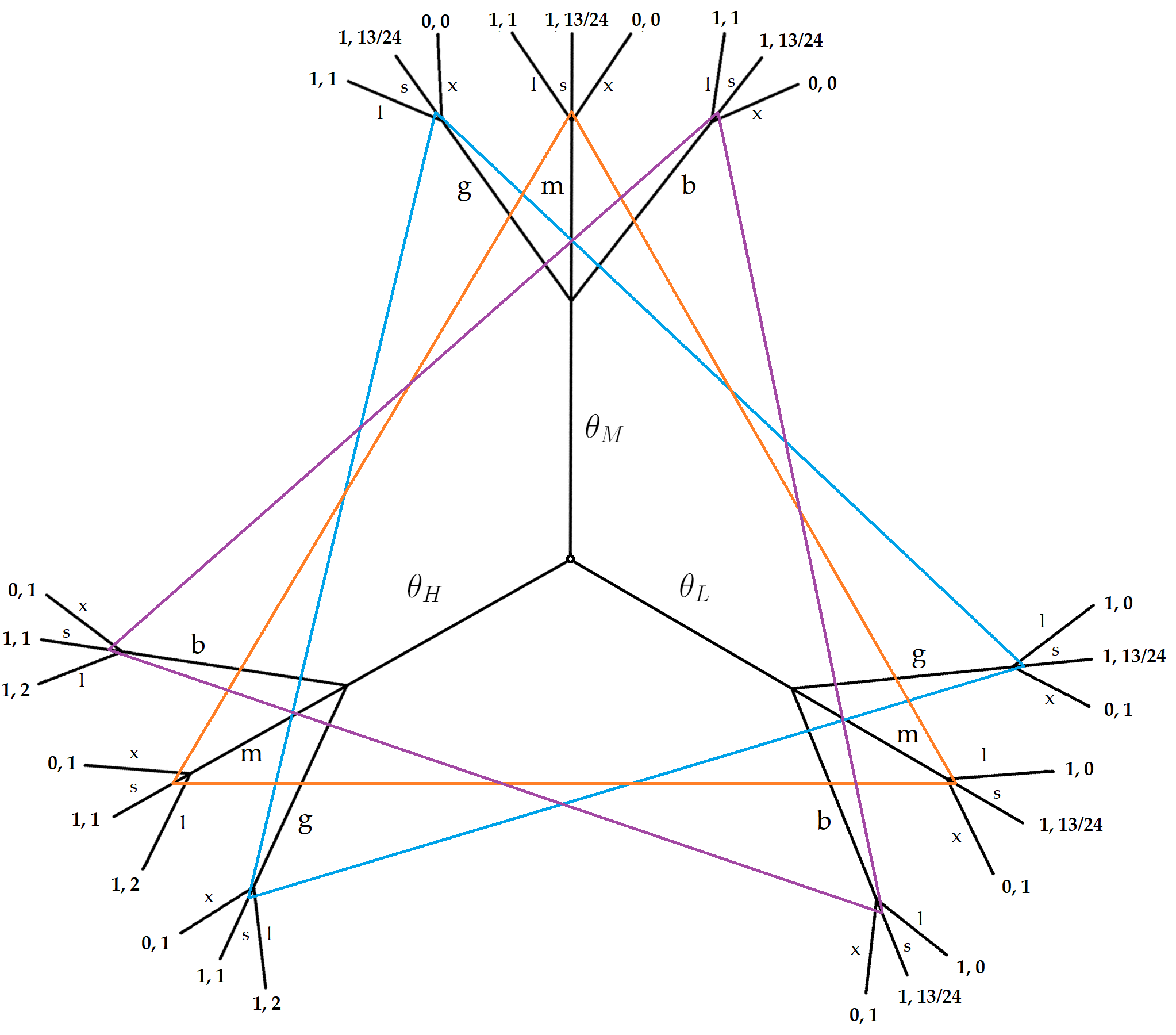}
\end{center}
\caption{\label{convexfig} Lemma \ref{3by3} Game}\end{figure}
 There are three types $\theta_{L}$, $\theta_{M}$, and $\theta_{H}$. Write a belief as a triple $(\mu_{L}, \mu_{M},\mu_{H})$. Note that the game is cheap talk with transparent motives: each type gets utility $1$ if the receiver chooses $l$ or $s$, and $0$ if the receiver chooses $x$. The receiver's preferences are given as follows:
 
 \begin{center}
\begin{tabular}{ |c|c|c|c| } 
\hline
Action & $\theta_{L}$ & $\theta_{M}$ & $\theta_{H}$ \\ 
$l$ & $0$  & $1$ & $2$ \\ 
$s$  & $13/24$ & $13/24$ & $1$ \\ 
$x$  & $1$ & $0$ & $1$ \\ 
\hline
\end{tabular}
\end{center}
 
 Consider the following three beliefs
\[\begin{split}
    \mu_{0} &\coloneqq \left(\frac{1}{4}, \frac{1}{4}, \frac{1}{2}\right), \qquad \mu_{1} \coloneqq \left(\frac{1}{12},\frac{1}{4}, \frac{2}{3}\right), \quad \text{and} \quad \mu_{2} \coloneqq \left(\frac{5}{12},\frac{1}{4}, \frac{1}{3}\right)
\end{split}\]
and note that $\mu_{0}$ is a convex combination of $\mu_{1}$ and $\mu_{2}$, each with weight $1/2$. That is, for some prior $\mu_{0}$, $\mu_{1}$ and $\mu_{2}$ are the realizations of a binary initial experiment, $\zeta$.

We depict the three prior distributions in the $(x,y)$-coordinate plane, where the $x$-axis corresponds to $\mu_{H}$, and the $y$-axis corresponds to $\mu_{M}$. There exist three convex regions of beliefs, $l$, $s$, and $x$, in which actions $l$, $s$, and $x$, respectively, are optimal. These regions and the three beliefs, $\mu_{0}$, $\mu_{1}$, and $\mu_{2}$, are illustrated in Figure \ref{convex1}.

\begin{figure}
\begin{center}
\includegraphics[scale=.7]{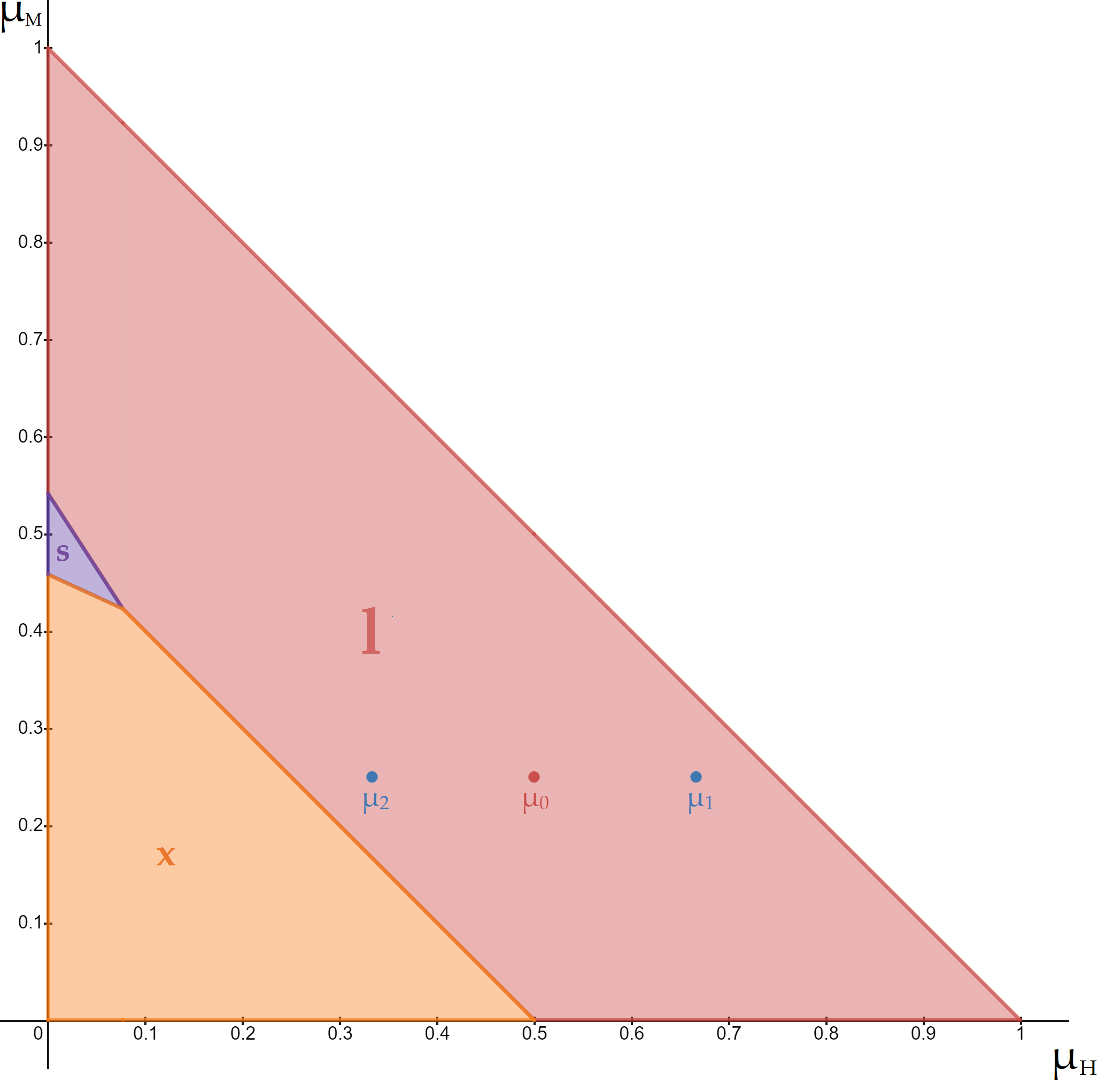}
\end{center}
\caption{\label{convex1} Lemma \ref{3by3} Action Regions}\end{figure}

After some effort (relegated to Appendix \ref{3by3proof}), we conclude that at beliefs $\mu_{0}$ and $\mu_{1}$ the receiver-optimal equilibrium is one in which $\theta_{H}$ and $\theta_{L}$ choose different messages, say $g$ and $b$, respectively; and $\theta_{M}$ mixes between those messages ($g$ and $b$). The receiver's payoffs at these beliefs are $67/52$ and $83/52$, respectively.

In contrast, at belief $\mu_{2}$, such an equilibrium does not exist. Instead, the receiver-optimal equilibrium sees $\theta_{H}$ and $\theta_{M}$ choose different messages, say $g$ and $m$, respectively; and $\theta_{L}$ mix between those messages ($g$ and $m$). The receiver's payoff is $127/132$.

The posteriors corresponding to the y-equilibrium experiments for $\mu_{1}$ and $\mu_{2}$ and the null-optimal experiment are depicted in Figure \ref{convex2}, where each $x$ denotes a posterior distribution.

\begin{figure}
\begin{center}
\includegraphics[scale=.7]{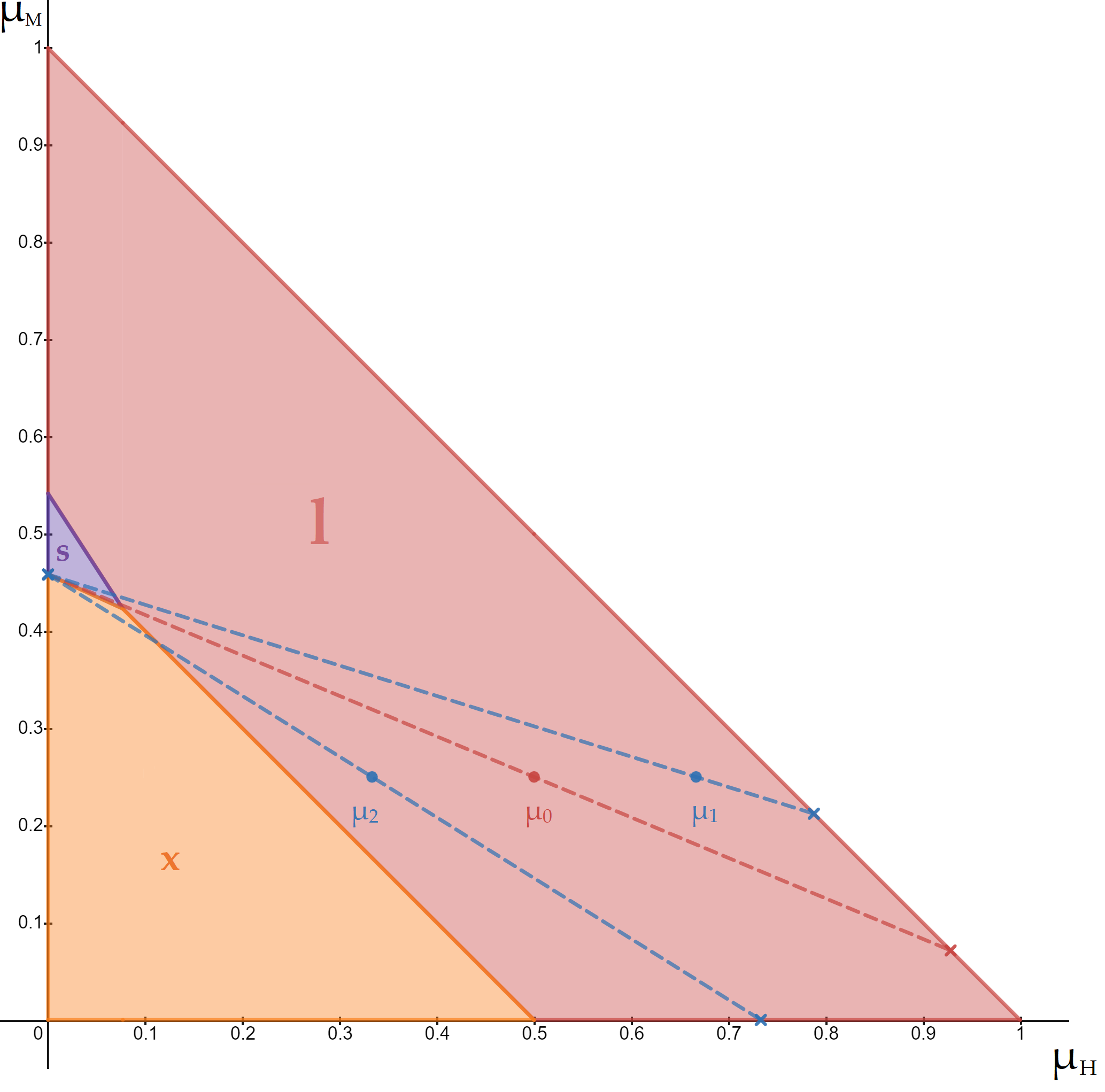}
\end{center}
\caption{\label{convex2} Lemma \ref{3by3} Optimal Posteriors}\end{figure}

Finally, we can directly calculate and compare the receiver's expected payoff from this information acquisition to her payoff without obtaining the information:
\[\frac{67}{52} > \frac{2195}{1716} = \frac{1}{2}\cdot \frac{83}{52} + \frac{1}{2}\cdot \frac{127}{132}\]
whence we conclude that the receiver's optimal equilibrium payoff is not convex in the prior. Note that this game is a cheap talk game with transparent motives--even these restrictions are not enough to guarantee convexity. 
\end{proof}

An analog to Figure \ref{case1} is depicted in Figure \ref{big}. There, the red point corresponds to the prior $\mu_{0}$, the blue arrows and points to the initial experiment and posteriors, the yellow arrows and points to the null-optimal experiment, the green arrows to the y-equilibrium experiments, and the purple arrows and points to experiment $\xi$.

\begin{figure}
\begin{center}
\includegraphics[scale=.7]{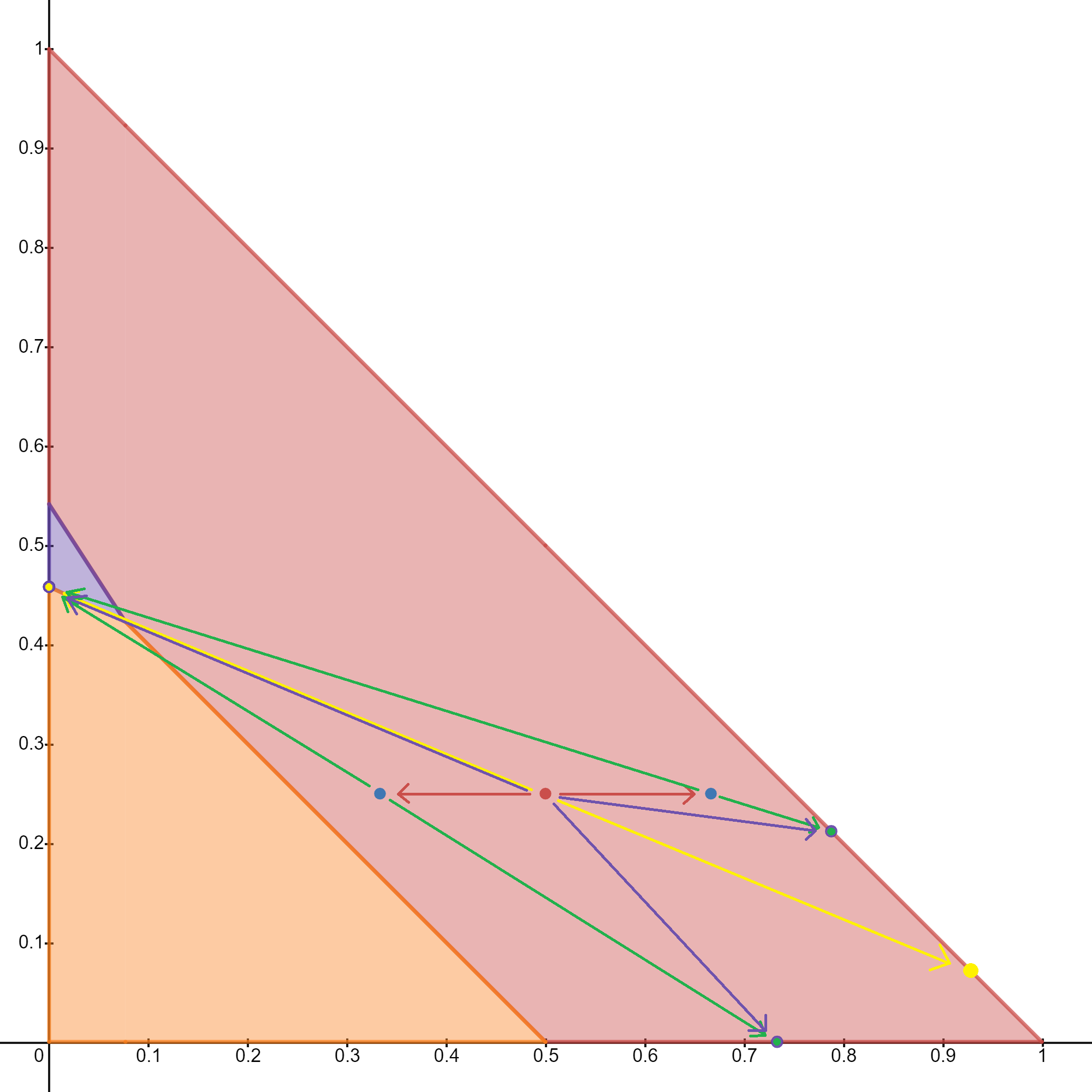}
\end{center}
\caption{\label{big} Lemma \ref{3by3} Induced Experiments}\end{figure}

What goes wrong when there are three or more states and actions? Recall the two state case. In such a setting, because there are only two states, the actors' beliefs are one dimensional. As a result, any additional information can only shift the belief to the left or right on the one-dimensional simplex of beliefs. Moreover, because of the of lack of diversity of sender types, the set of possible equilibrium vectors of strategies is quite small (qualitatively)--they either pool, separate, both mix, or only one mixes. Consequently, a change in the prior cannot have too great of an effect on the resulting equilibrium distribution of posteriors: as long as the change is in the correct direction (and remember, there are only two possible directions) and/or is sufficiently small, the receiver-optimal equilibrium from the original prior remains feasible, and hence the same vector of posteriors can be generated at equilibrium (albeit with different probabilities).

Furthermore, any change in the prior that eliminates the receiver-optimal equilibrium from the original prior must be large, so large that the resulting belief is more extreme than any posterior generated by the original equilibrium. Thus, the beliefs that correspond to $\xi$ must be the same, or more extreme than the beliefs from $\eta$, and hence learning must be beneficial. Put another way, there is a trade-off to initial learning--the gain in information from the initial experiment versus the (possibly decreased) gain in information from the receiver-optimal equilibria at the new priors. With two states, the first effect dominates, and makes up for the fact that the gain in information from the equilibria may be diminished.

With more than two states, this is no longer true. As in the two state case, initial learning can result in priors for the communication game for which the receiver-optimal equilibrium under the original prior is no longer feasible. However, due to the fact that the belief space is now multi-dimensional, the beliefs generated by the receiver-optimal equilibrium at these new priors, while more extreme in \textit{some} direction, do not correspond in general to a more informative experiment. Hence, $\xi$ and $\eta$ may not be Blackwell comparable, in which case comparisons of the receiver welfare between the no learning and learning scenarios must rely on the specific details of the initial experiment and payoffs of the game.  As in the two state case, there is the same trade-off to initial learning, but now the initial gain may not dominate.

In the counterexample constructed in Lemma \ref{3by3}, the proposed initial experiment is harmful, since it involves too much learning about whether the state is $\theta_{L}$. In particular, for belief $\mu_{2}$, the receiver is too confident that the state is $\theta_{L}$, which precludes the existence of an equilibrium in which the receiver can distinguish between the high type and the low type. Instead, the receiver-optimal equilibrium is one in which she can distinguish between the high type and the medium type, which is much less helpful for the receiver. 

As discussed in Whitmeyer (2019) \cite{Whit}, the value of information may not be positive in this game even when the receiver can choose the optimal information structure in the signaling game. As we discover in Whitmeyer (2019), the equilibria described above at beliefs $\mu_{0}$, $\mu_{1}$, $\mu_{2}$, yield the maximum payoffs to the receiver of any equilibrium under any information structure.

Finally, if there are only two states and actions, the receiver's payoff may fail to be convex in the prior if the game is not simple. To wit,

\begin{lemma}\label{ls}
If the communication game is not simple, then the receiver's payoff is not generally convex in the prior.
\end{lemma}
\begin{proof}
Proof is via counter example. Consider the modified Beer-Quiche game (\textit{cf.} Cho and Kreps (1987) \cite{cho}) depicted in Figure \ref{nonsimpcount}, in which the receiver now obtains an additional payoff of $1$ if the sender chooses Quiche and the receiver chooses the ``correct" action (i.e. $F$ if the sender is $\theta_{W}$ and $NF$ if the sender is $\theta_{S}$). 

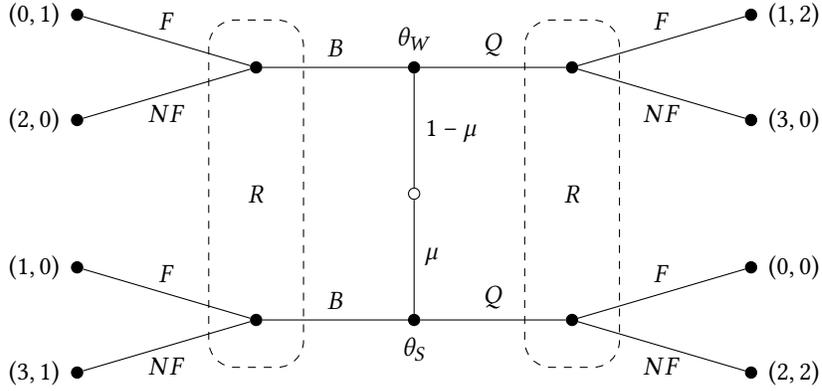
\begin{figure}
    \centering
    \begin{tikzpicture}[scale=1.4,font=\footnotesize]
\tikzset{
solid node/.style={circle,draw,inner sep=1.5,fill=black},
hollow node/.style={circle,draw,inner sep=1.5}}
\tikzstyle{level 1}=[level distance=12mm,sibling distance=25mm]
\tikzstyle{level 2}=[level distance=15mm,sibling distance=15mm]
\tikzstyle{level 3}=[level distance=17mm,sibling distance=10mm]
\node(0)[hollow node]{}
child[grow=up]{node[solid node,label=above:{$\theta_{W}$}] {}
child[grow=left]{node(1)[solid node]{}
child{node[solid node,label=left:{$(0, 1)$}]{} edge from parent node [above]{$F$}}
child{node[solid node,label=left:{$(2, 0)$}]{} edge from parent node [below]{$NF$}}
edge from parent node [above]{$B$}}
child[grow=right]{node(3)[solid node]{}
child{node[solid node,label=right:{$(3, 0)$}]{} edge from parent node [below]{$NF$}}
child{node[solid node,label=right:{$(1, 2)$}]{} edge from parent node [above]{$F$}}
edge from parent node [above]{$Q$}}
edge from parent node [right]{$1-\mu$}}
child[grow=down]{node[solid node,label=below:{$\theta_{S}$}] {}
child[grow=left]{node(2)[solid node]{}
child{node[solid node,label=left:{$(1, 0)$}]{} edge from parent node [above]{$F$}}
child{node[solid node,label=left:{$(3, 1)$}]{} edge from parent node [below]{$NF$}}
edge from parent node [above]{$B$}}
child[grow=right]{node(4)[solid node]{}
child{node[solid node,label=right:{$(2, 2)$}]{} edge from parent node [below]{$NF$}}
child{node[solid node,label=right:{$(0, 0)$}]{} edge from parent node [above]{$F$}}
edge from parent node [above]{$Q$}}
edge from parent node [right]{$\mu$}};
\draw[dashed,rounded corners=10]($(1) + (-.45,.45)$)rectangle($(2) +(.45,-.45)$);
\draw[dashed,rounded corners=10]($(3) + (-.45,.45)$)rectangle($(4) +(.45,-.45)$);
\node at ($(1)!.5!(2)$) {$R$};
\node at ($(3)!.5!(4)$) {$R$};
\end{tikzpicture}
    \caption{Lemma \ref{ls} Game}
    \label{nonsimpcount}
\end{figure}

If $\mu \geq 1/2$, the receiver optimal equilibrium is one in which the sender types pool on $Q$. The receiver's best response is $NF$ and his expected payoff is $2\mu$. If $\mu < 1/2$, the receiver optimal equilibrium is one in which $\theta_{W}$ mixes, choosing $B$ with probability $\sigma = \mu/(1-\mu)$, and $\theta_{S}$ chooses $B$. The receiver's expected payoff for $\mu < 1/2$ is thus $\Pr(B)/2 + 2\Pr(Q) = 2 - 3\mu$. Hence,

\[V^{T}(\mu) = \begin{cases}
2-3\mu, & \mu < \frac{1}{2}\\
2 \mu, & \mu \geq \frac{1}{2}
\end{cases}\]

Figure \ref{nonsimp} depicts $V^{T}$. The receiver's payoff is neither convex nor lower semi-continuous. The dotted line in the figure is a secant line that corresponds to a binary initial experiment that strictly hurts the receiver.
\begin{figure}
\begin{center}
\includegraphics[scale=.7]{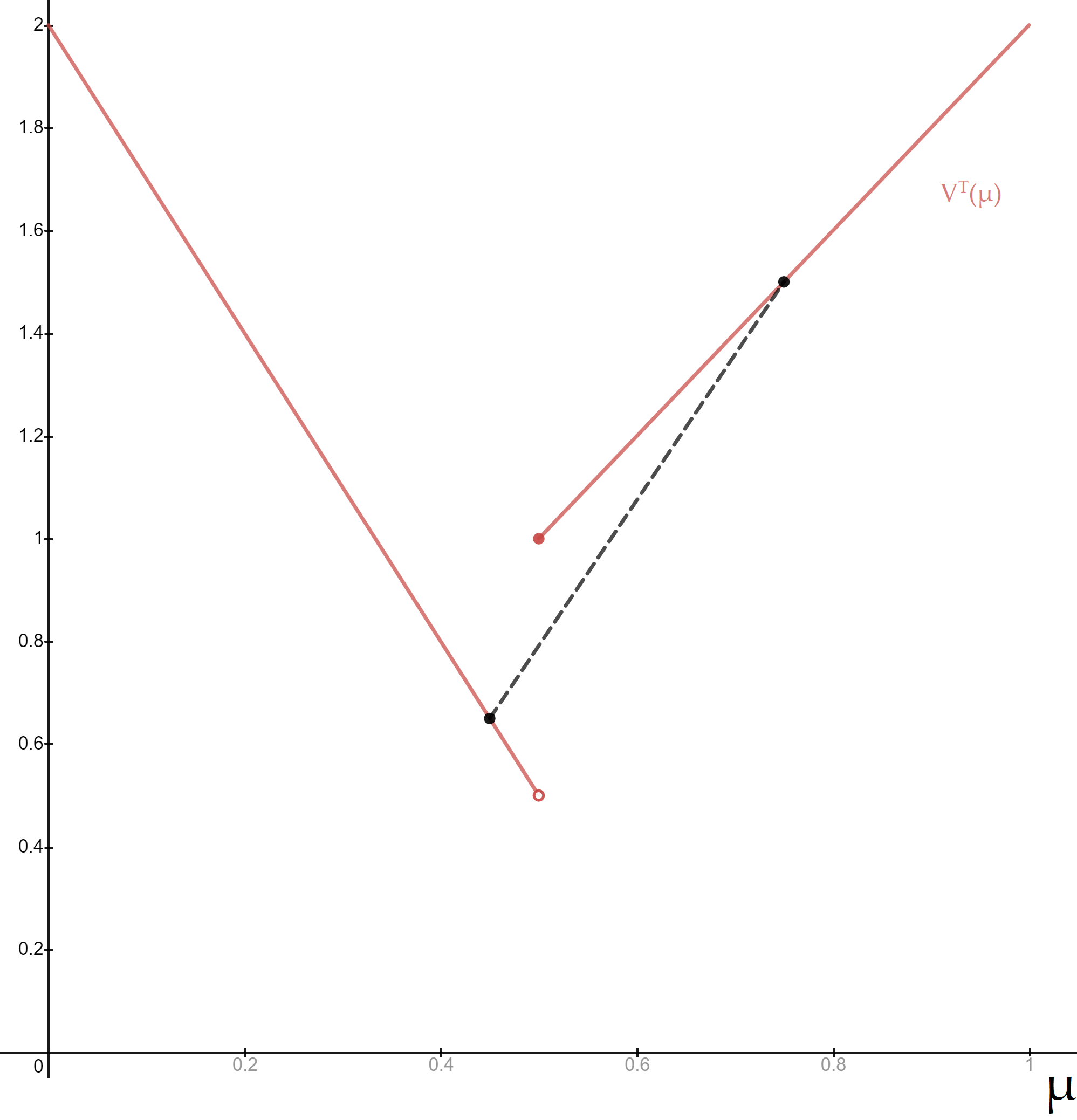}
\end{center}
\caption{\label{nonsimp} Receiver Payoffs (Lemma \ref{ls} Proof)}\end{figure}

\end{proof}

\section{Conclusion}\label{conclusion}

This paper comprehensively answers the question of when information always benefits the receiver in two player communication games. As Theorem \ref{main2} illustrates, Theorem \ref{main1} is as strong as possible--should none of its conditions hold, the value of information may be strictly negative.

Naturally, there is room for more work on related questions. What can we say, for instance, about the value of information for the sender? Answering such a question would pose a challenge since the proof techniques used in this paper would no longer work. Here we are able to bypass the details of the sender's incentives and work with distributions of beliefs. This allows us to tackle the problem as a decision problem for the receiver, in which we apply Blackwell's theorem. Such an approach would not work when exploring sender welfare because he is not a decision maker.

\bibliography{sample.bib}
\appendix
\section{Proofs}
\subsection{Lemma \ref{prop22} Proof}\label{22proof}
\begin{proof}

If there exists a separating equilibrium in the game then the result is trivial; henceforth we consider only simple signaling games that do not admit separating equilibria. Since there are just two states of the world, $\Theta \coloneqq \left\{\theta_{L},\theta_{H}\right\}$, any belief (probability distribution over states) can be completely characterized by the parameter $\mu \coloneqq \Pr(\Theta = \theta_{H})$. The interval of beliefs can be partitioned into finitely many partitions with boundaries $0 < \mu_{1} < \mu_{2} < \cdots < \mu_{k} < 1$. If an action is optimal for a receiver given some belief, $\mu$, then it is either optimal only at that belief or it is optimal for a closed interval of beliefs $\left[\mu_{i}, \mu_{j}\right]$ of which $\mu$ is a member. 

We first show that it is without loss of generality to restrict the sender to two messages.

\begin{claim}\label{steak}
For any equilibrium that yields the receiver a payoff of $v$ in which $m > 2$ messages are used, there exists an equilibrium in which at most $2$ messages are used that yields the receiver a payoff that is weakly higher than $v$.
\end{claim}

\begin{proof}
Suppose that $l > 2$ messages are used. Since there exist no separating equilibria, there are only two feasible $l$-message equilibria: i. Both types choose a mixed strategy with full support, or ii. One type (say $\theta_{H}$) chooses a mixed strategy with full support, and the other type chooses a mixed strategy with support on all but one message. 

In both cases, there will be $l$ resulting equilibrium beliefs $\mu_{1}' < \mu_{2}' < \cdots < \mu_{l}'$, where in case ii. $\mu_{l}' = 1$. However since each type is mixing, they must be indifferent over each message in the support of their mixed strategy. Hence, there must also be an equilibrium in both cases in which only two messages are used, which induce beliefs $\mu_{1}'$ and $\mu_{l}'$. Indeed such an equilibrium can be constructed by taking each on-path message $m_{i}$ with the associated induced belief $\mu_{i}'$, with $i \neq 1, l$, and moving weight from each player's mixed strategy on $m_{i}$ to message $m_{l}$ at the ratio

\[\frac{\Delta\left(\sigma_{H}(m_{i})\right)}{\Delta\left(\sigma_{L}(m_{i})\right)} = \frac{\Delta\left(\sigma_{H}(m_{l})\right)}{\Delta\left(\sigma_{L}(m_{l})\right)} = \frac{(1-\mu_{0})\mu_{l}'}{(1-\mu_{l}')\mu_{0}}\]
Such a process decreases $\mu_{i}'$ and by construction maintains $\mu_{l}'$. This can be done until $\mu_{i}' = \mu_{1}'$ for each $i$.

The Blackwell experiment that corresponds to this new, binary, distribution of posteriors is more informative than in the original situation, where $l$ messages were used. Hence, the receiver's payoff must be weakly higher in the two-message equilibrium.
\end{proof}

Thus, suppose that there are just two messages in the game. There are three cases to consider, each of which is illustrated in Figure \ref{case1}.

\textbf{Case 1:} For prior $\mu_{0}$, the receiver-optimal equilibrium is one in which both types pool. Observe that in this case, the null-optimal experiment is a completely uninformative experiment.

Consider any initial experiment $\zeta$ with $k \geq 2$ realizations. Following this experiment, there are $k$ posteriors, $\mu_{y}$, $y \in \left\{1,\dots,k\right\}$; or equivalently there are $k$ priors in the resulting signaling game.

For each realization of the initial experiment, $y$, the receiver's equilibrium payoff in the resulting signaling game is clearly bounded below by the payoff from a pooling equilibrium, since that corresponds to the least informative (in the Blackwell sense) y-equilibrium experiment. Hence, suppose that in each of the $k$ signaling games, there is a pooling equilibrium, and that is the receiver-optimal equilibrium.

Finally, define $\xi$ as the experiment that corresponds to the information ultimately acquired by the receiver following the initial learning and the resulting equilibrium play in the signaling game. The null-optimal experiment, $\eta$, is less Blackwell informative than $\xi$ and so the receiver prefers $\xi$--the receiver prefers any learning.

\textbf{Case 2:} For prior $\mu_{0}$, the receiver optimal equilibrium is one in which one type mixes and the other type chooses a pure strategy. Observe that in this case, the null-optimal experiment begets two posteriors: one that is in the interior on $[0,1]$ and the other that is either $0$ or $1$.

Without loss of generality (the other cases follow analogously) suppose that $\theta_{H}$ mixes and chooses message $m_1$ with probability $\sigma$ and $\theta_{L}$ chooses message $m_1$. Following an observation of message $m_2$, the receiver's belief is $1$ and following message $m_1$ it is $\mu_{j} < \mu_{0}$. Moreover, using Bayes' law we obtain \[\sigma = \frac{(1-\mu_{0})\mu_{j}}{(1-\mu_{j})\mu_{0}} \]

Consider any initial experiment, $\zeta$, with $k \geq 2$ realizations. Observe that for any realization that yields a belief $\mu_{i} \geq \mu_{j}$, an equilibrium in which $\theta_{H}$ mixes and $\theta_{L}$ does not must also exist, and hence the receiver's equilibrium payoffs for each of these beliefs must be bounded below by the payoff for that equilibrium. As in case $1$, suppose that in each case that this equilibrium is optimal (and hence that the receiver's payoffs are at their lower bounds).

For each $y$ such that $\mu_{y} \geq \mu_{j}$, the y-equilibrium experiment is one that sends the posteriors to $\mu_{j}$ and $1$. Consequently, it is without loss of generality to suppose that $\zeta$ just has a single experiment realization that yields a belief above $\mu_{j}$. Moreover, as in the first case, any realization of experiment  $\zeta$ that yields a posterior $\mu_{y} < \mu_{j}$ must beget an equilibrium payoff bounded below by the pooling payoff. Hence, we suppose that for each such realization $y$, the optimal equilibrium is the pooling equilibrium. Moreover, the resulting payoff from this distribution over pooling payoffs itself is bounded below by the payoff were the initial experiment to have merely a single signal $y$ that begets a belief below $\mu_{j}$. Accordingly, we suppose that is the case.

To summarize, $\zeta$ has just two signal realizations, $y_{1}$ and $y_{2}$, corresponding to beliefs $\mu_{1} < \mu_{j}$ and $\mu_{2} > \mu_{j}$, respectively. $\gamma_{1}$ has just one signal realization, corresponding to belief $\mu_{1}$. $\gamma_{2}$ has two signal realizations, corresponding to beliefs $\mu_{j}$ and $1$. Hence, $\xi$ has three signal realizations, corresponding to beliefs $\mu_{1}, \mu_{j}$ and $1$. The null-optimal experiment $\eta$ has two signal realizations, corresponding to beliefs $\mu_{j}$ and $1$.

The resulting distribution over posteriors induced by $\eta$ has support on $1$ and $\mu_{j}$. Likewise, the resulting distribution over posteriors induced by $\xi$ has support on $1$, $\mu_{j}$ and $\mu_{1}$. Since $\mu_{1} < \mu_{j}$, $\xi$ is more Blackwell informative than $\eta$ and so the receiver prefers $\xi$--the receiver prefers learning.

\textbf{Case 3:} For prior $\mu_{0}$, the receiver optimal equilibrium is one in which both types mix. Observe that in this case, the null-optimal experiment begets two posteriors, both of which are in the interior of $[0,1]$.

Let the high type choose a mixed strategy $\sigma_{H}$ and let the low type choose a mixed strategy $\sigma_{L}$. The receiver will have two posteriors, $\mu_{j} > \mu_{0} > \mu_{l}$ and using Bayes' law, we have

\[\sigma_{H} = \frac{\mu_{j}\left(\mu_{0}-\mu_{l}\right)}{\mu_{0}\left(\mu_{j}-\mu_{l}\right)}, \quad \text{and} \quad \sigma_{L} = \frac{\left(1-\mu_{j}\right)\left(\mu_{0}-\mu_{l}\right)}{\left(1-\mu_{0}\right)\left(\mu_{j}-\mu_{l}\right)} \]
The remainder proceeds in the same way as in the first two cases, any experiment realization that yields a belief in the interval $[\mu_{l},\mu_{j}]$ leads to an equilibrium payoff bounded below by the optimal equilibrium payoff at belief $\mu_{0}$, and any experiment realization that yields a belief outside that interval leads to an equilibrium payoff bounded below by the pooling payoff. 

Ultimately, the null-optimal experiment, $\eta$, is less Blackwell informative than $\xi$ (which corresponds to an information structure that serves as a \textit{lower-bound} for the receiver's payoff from learning), so learning must always be beneficial. 

We have gone through each case, and the result is shown.
\end{proof}

\subsection{Lemma \ref{nodivide} Proof}\label{ndproof}
\begin{proof}

Let each action be strictly optimal in at least one state (or else the result is trivial). We may partition the set of types $\Theta = \Theta_{1} \sqcup \Theta_{2}$, where $\Theta_{1}$ is the set of types for whom the receiver strictly prefers to choose action $a_{1}$, and $\Theta_{2}$ is the set of types for whom the receiver strictly prefers to choose action $a_{2}$. It is without loss of generality to suppose that there are no types for whom the receiver is indifferent between her two actions.

Equivalently, $v_{i} \coloneqq u_{R}(a_{1},\theta_{i}) > u_{R}(a_{2},\theta_{i}) \eqqcolon w_{i}$ for all $\theta_{i} \in \Theta_{1}$, and $w_{i} > v_{i}$ for all $\theta_{i} \in \Theta_{2}$. Denote
\[\begin{split}
    \Theta_{1} &= \left\{\theta_{1},\dots,\theta_{t}\right\}\\
    \Theta_{2} &= \left\{\theta_{t+1},\dots,\theta_{n}\right\}\\
\end{split}\]
Consider an equilibrium in which at least one type, $\theta_{k}$, mixes. Without loss of generality, let $\theta_{k} \in \Theta_{1}$, i.e. he is a type for whom the receiver would strictly prefer to choose action $a_1$. Next, suppose that $\theta_{k}$ mixes over a subset of the set of messages, $M_{k}$, where

\[M_{k} = \left\{m_{1},\dots, m_{l}\right\}\] with a generic element $m_{k} \in M_{k}$. Moreover, $M_{k}$ is partitioned by three sets, $M_{k}^{0}$, $M_{k}^{1}$ and $M_{k}^{2}$, where $M_{k}^{0}$ is the set of messages after which the receiver is indifferent between her two actions, $M_{k}^{1}$ is the set of messages after which the receiver strictly prefers $a_{1}$, and $M_{k}^{2}$ is the set of messages after which the receiver strictly prefers $a_{2}$. By assumption, neither $M_{k}^{1}$ nor $M_{k}^{2}$ is the empty set, in which case we can pick two messages, $m_{1} \in M_{k}^{1}$ and $m_{2} \in M_{k}^{2}$. The receiver's expected payoff at this equilibrium can be written as 
\[V(\mu) = \mu\left(\theta_{k}\right)\left(\sigma_{k}\left(m_{1}\right)v_{k} + \sigma_{k}\left(m_{2}\right)w_{k}\right) + \gamma\]
where $\gamma$ is the remainder of the receiver's payoff that--crucially for the sake of this proof--does not depend on $\sigma_{k}\left(m_{1}\right)$ or $\sigma_{k}\left(m_{2}\right)$. However, the receiver's payoff strictly increases if instead $\theta_{k}$ modified his mixed strategy so that $\hat{\sigma}_{k}\left(m_{1}\right) = \sigma_{k}\left(m_{1}\right) + \sigma_{k}\left(m_{2}\right)$ and $\hat{\sigma}_{k}\left(m_{2}\right) = 0$ since $v_{k} > w_{k}$ (recall that we stipulated that $\theta_{k} \in \Theta_{1}$). Moreover, it is easy to see that this is also an equilibrium: $\theta_{k}$ is indifferent over any pure strategy in the support of his mixed strategy; and following messages $m_{1}$ and $m_{2}$ under the new mixture, the receiver still finds it optimal to choose $a_{1}$ and $a_{2}$, respectively (and the receiver's beliefs and payoffs following any other message are unchanged).

Since $m_{1} \in M_{k}^{1}$ and $m_{2} \in M_{k}^{2}$ were two arbitrary messages, and $\theta_{k}$ was an arbitrary type, the result follows.
\end{proof}

\subsection{Lemma \ref{too} Proof}\label{tooproof}
\begin{proof}
Again, let each action be uniquely optimal in at least one state, and let two messages be used in the receiver optimal equilibrium at belief $\mu_{0}$ (if only one message is used, the receiver obtains the pooling payoff at $\mu_{0}$, and hence any initial experiment must be to her profit).

In addition, we may, without loss of generality, impose that at $\mu_{0}$ there is an equilibrium such that, following each message, $m_{1}$ and $m_{2}$, different actions, $a_{1}$ and $a_{2}$, respectively, are strictly optimal. Otherwise, this would just yield the pooling payoff and the result would be trivial. By Lemma \ref{nodivide}, this imposition ensures that each type is choosing a pure strategy.

Next, partition set $\Theta_{1}$ (define in Appendix \ref{ndproof}) into two sets, $G_{1}^{e}$ and $G_{1}^{d}$, which correspond to the types who choose $m_{1}$ and $m_{2}$, respectively. Likewise, partition set $\Theta_{2}$ into two sets $G_{2}^{e}$ and $G_{2}^{d}$, which correspond to the types who choose $m_{2}$ and $m_{1}$, respectively.

These sets are, explicitly,

\[\begin{split}
    G_{1}^{e} &\coloneqq \left\{\theta_{1},\dots,\theta_{m}\right\}, \qquad G_{1}^{d} \coloneqq \left\{\theta_{m+1},\dots,\theta_{t}\right\}\\
    G_{2}^{e} &\coloneqq \left\{\theta_{t+1},\dots,\theta_{r}\right\}, \qquad G_{2}^{d} \coloneqq \left\{\theta_{r+1},\dots,\theta_{n}\right\}\\
\end{split}\]

Moreover, note that it is possible that some are the empty set; although of course if $G_{2}^{d}$ is nonempty then $G_{1}^{e}$ cannot be empty, and similarly for $G_{1}^{d}$ and $G_{2}^{e}$. 

Next, since we have imposed that an equilibrium of the above form exists at $\mu_{0}$, such an equilibrium must exist at any belief $\mu$ such that the following condition holds:

\begin{condition}\label{cond44}
\[\tag{$A1$}\label{100}\sum_{i=1}^{m}\mu\left(\theta_{i}\right)v_{i} + \sum_{j = r+1}^{n}\mu\left(\theta_{j}\right)v_{j} \geq \sum_{i=1}^{m}\mu\left(\theta_{i}\right)w_{i} + \sum_{j = r+1}^{n}\mu\left(\theta_{j}\right)w_{j}\]
and
\[\tag{$A2$}\label{101}\sum_{i=t+1}^{r}\mu\left(\theta_{i}\right)w_{i} + \sum_{j = m+1}^{t}\mu\left(\theta_{j}\right)w_{j} \geq \sum_{i=t+1}^{r}\mu\left(\theta_{i}\right)v_{i} + \sum_{j = m+1}^{t}\mu\left(\theta_{j}\right)v_{j}\]
\end{condition}
Accordingly, at $\mu_{0}$, the receiver's payoff is
\[\label{103}\tag{$A3$}\begin{split}
    V^{T}(\mu_{0}) = \sum_{i=1}^{m}\mu_{0}\left(\theta_{i}\right)v_{i} &+ \sum_{j = r+1}^{n}\mu_{0}\left(\theta_{j}\right)v_{j} +
\sum_{i=t+1}^{r}\mu_{0}\left(\theta_{i}\right)w_{i} + \sum_{j = m+1}^{t}\mu_{0}\left(\theta_{j}\right)w_{j}
\end{split}\]

Without loss of generality, we may suppose that there are just three signal realizations: one, $A$, after which Condition \ref{cond44} holds; one, $B$, after which there is a pooling equilibrium for which $a_{1}$ is optimal; and one, $C$, after which there is a pooling equilibrium for which $a_{2}$ is optimal. We may make this assumption since the receiver is only aided by multiple signal realizations in each pooling region. Then the receiver's expected payoff from the initial signal is bounded below by

\[\tag{$A4$}\label{104}\begin{split}
     p&\left[\sum_{i=1}^{m}\mu_{A}\left(\theta_{i}\right)v_{i} + \sum_{j = r+1}^{n}\mu_{A}\left(\theta_{j}\right)v_{j} + \sum_{i=t+1}^{r}\mu_{A}\left(\theta_{i}\right)w_{i} + \sum_{j = m+1}^{t}\mu_{A}\left(\theta_{j}\right)w_{j}\right]\\
     + &q\left[\sum_{i=1}^{m}\mu_{B}\left(\theta_{i}\right)v_{i} + \sum_{j = r+1}^{n}\mu_{B}\left(\theta_{j}\right)v_{j} + \sum_{i=t+1}^{r}\mu_{B}\left(\theta_{i}\right)v_{i} + \sum_{j = m+1}^{t}\mu_{B}\left(\theta_{j}\right)v_{j}\right]\\
     &+ s\left[\sum_{i=1}^{m}\mu_{C}\left(\theta_{i}\right)w_{i} + \sum_{j = r+1}^{n}\mu_{C}\left(\theta_{j}\right)w_{j} + \sum_{i=t+1}^{r}\mu_{C}\left(\theta_{i}\right)w_{i} + \sum_{j = m+1}^{t}\mu_{C}\left(\theta_{j}\right)w_{j}\right]
\end{split}\]
where
\[p \coloneqq \Pr(A), \quad q \coloneqq \Pr(B), \quad s \coloneqq \Pr(C), \quad p + q + s = 1,\]
and $p\mu_{A}(\theta_{i}) + q\mu_{B}(\theta_{i}) + r\mu_{C}(\theta_{i}) = \mu_{0}(\theta_{i})$ for all $i$. Expression \ref{104} can be simplified to
\[\begin{split}
    V^{T}(\mu_{0}) + q\Upsilon + s\Gamma
\end{split}\]
where
\[\tag{$A5$}\label{105}\begin{split}
    \Upsilon \coloneqq \left[\sum_{i=t+1}^{r}\mu_{B}\left(\theta_{i}\right)\left(v_{i}-w_{i}\right) + \sum_{j = m+1}^{t}\mu_{B}\left(\theta_{j}\right)\left(v_{j} - w_{j}\right)\right]
\end{split}\]
and
\[\tag{$A6$}\label{106}\begin{split}
    \Gamma \coloneqq \left[\sum_{i=1}^{m}\mu_{C}\left(\theta_{i}\right)\left(w_{i}-v_{i}\right) + \sum_{j = r+1}^{n}\mu_{C}\left(\theta_{j}\right)\left(w_{j}-v_{j}\right)\right]
\end{split}\]
Since $a_{1}$ is optimal in the pooling equilibrium following message $B$ and $a_{2}$ is optimal in the pooling equilibrium following $C$ we must have
\[\tag{$A7$}\label{107}\begin{split}
    \sum_{i=1}^{m}\mu_{B}\left(\theta_{i}\right)\left(v_{i}-w_{i}\right) + \sum_{j = r+1}^{n}\mu_{B}\left(\theta_{j}\right)\left(v_{j}-w_{j}\right) \geq \sum_{i=t+1}^{r}\mu_{B}\left(\theta_{i}\right)\left(w_{i}-v_{i}\right) + \sum_{j = m+1}^{t}\mu_{B}\left(\theta_{j}\right)\left(w_{j}-v_{j}\right)
\end{split}\]
and

\[\tag{$A8$}\label{108}\begin{split}
    \sum_{i=1}^{m}\mu_{C}\left(\theta_{i}\right)\left(v_{i}-w_{i}\right) + \sum_{j = r+1}^{n}\mu_{C}\left(\theta_{j}\right)\left(v_{j}-w_{j}\right) \leq \sum_{i=t+1}^{r}\mu_{C}\left(\theta_{i}\right)\left(w_{i}-v_{i}\right) + \sum_{j = m+1}^{t}\mu_{C}\left(\theta_{j}\right)\left(w_{j}-v_{j}\right)
\end{split}\]
Moreover, since Condition \ref{cond44} does not hold for belief $\mu_{B}$, we must have either
\[\tag{$A9$}\label{109}\sum_{i=1}^{m}\mu_{B}\left(\theta_{i}\right)v_{i} + \sum_{j = r+1}^{n}\mu_{B}\left(\theta_{j}\right)v_{j} < \sum_{i=1}^{m}\mu_{B}\left(\theta_{i}\right)w_{i} + \sum_{j = r+1}^{n}\mu_{B}\left(\theta_{j}\right)w_{j}\]
and/or
\[\tag{$A10$}\label{110}\sum_{i=t+1}^{r}\mu_{B}\left(\theta_{i}\right)w_{i} + \sum_{j = m+1}^{t}\mu_{B}\left(\theta_{j}\right)w_{j} < \sum_{i=t+1}^{r}\mu_{B}\left(\theta_{i}\right)v_{i} + \sum_{j = m+1}^{t}\mu_{B}\left(\theta_{j}\right)v_{j}\]
Suppose that Inequality \ref{109} holds. then we may substitute it into Inequality \ref{107} and cancel:
\[\begin{split}
    \sum_{i=t+1}^{r}\mu_{B}\left(\theta_{i}\right)\left(v_{i}-w_{i}\right) + \sum_{j = m+1}^{t}\mu_{B}\left(\theta_{j}\right)\left(v_{j}-w_{j}\right) \geq 0
\end{split}\]
Hence, $\Upsilon$ is positive. On the other hand, if Inequality \ref{110} holds then we may substitute it directly into $\Upsilon$, which again must be positive.

A symmetric procedure works at belief $\mu_{C}$ to establish that $\Gamma$ also must be positive. Since $\Gamma$ and $\Upsilon$ are both positive, Expression \ref{104}, the receiver's payoff from learning must be at least weakly greater than $V^{T}(\mu_{0})$. We have exhausted every case, and so conclude that any initial experiment benefits the receiver.
\end{proof}

\subsection{Lemma \ref{409} Proof}\label{409proof}

\begin{proof}

Denote the set of (three) states by $\Theta = \left\{\theta_{L}, \theta_{M}, \theta_{H}\right\}$. For convenience we continue to use the shorthand $v_{i} \coloneqq u_{R}\left(a_{1},\theta_{i}\right)$ and $w_{i} \coloneqq u_{R}\left(a_{2},\theta_{i}\right)$, for all $i = L,M,H$. Without loss of generality, we may assume that action $a_{1}$ is strictly optimal in states $\theta_{L}$ and $\theta_{H}$, and action $a_{2}$ is strictly optimal in state $\theta_{M}$: $v_{L} > w_{L}$, $v_{H} > w_{H}$, and $w_{M} > v_{M}$.

Next, observe that we can picture any belief in the $(x,y)$-coordinate plane, where the $x$-axis corresponds to $\mu_{M}$, and the $y$-axis corresponds to $\mu_{H}$. Define region $R_{1}$ as the region in which $a_{1}$ is optimal and $R_{2}$ as the region in which $a_{2}$ is optimal. Each region, $R_{1}$ and $R_{2}$, is compact and convex, and the two regions share a boundary that is a line segment. Define $R$ to be the simplex of beliefs, $R = R_{1} \cup R_{2}$.

From Lemma \ref{nodivide}, for some prior $\mu_{0}$, there are just three possible arrangements of the posteriors that are induced by the null-optimal experiment: 

\vspace{.2cm}

\textbf{Case 1:} All posteriors lie in one region.

\vspace{.2cm}

\textbf{Case 2:} All posteriors that follow messages chosen by $\theta_{M}$ fall in region $R_{2}$, where there is at least one posterior that does not lie on the boundary $R_{1} \cap R_{2}$; and all posteriors that follow messages chosen by $\theta_{L}$ and $\theta_{H}$ fall in region $R_{1}$, where there is at least one posterior that does not lie on the boundary $R_{1} \cap R_{2}$.

\vspace{.2cm}

\textbf{Case 3:} All posteriors that follow messages chosen by $\theta_{H}$ ($\theta_{L}$) fall in region $R_{1}$, where there is at least one posterior that does not lie on the boundary $R_{1} \cap R_{2}$; and all posteriors that follow messages chosen by $\theta_{L}$ ($\theta_{H}$) and $\theta_{M}$ fall in region $R_{2}$, where there is at least one posterior that does not lie on the boundary $R_{1} \cap R_{2}$. By symmetry, we need focus only on the case where the posteriors that follow $\theta_{H}$'s messages fall in region $R_{1}$.

Note that throughout this proof, by Lemma \ref{nodivide}, each belief that is not on the line segment $R_{1} \cap R_{2}$ must lie on the boundary of the triangle ($2$-simplex) of beliefs.

In the first case, the receiver clearly benefits from any initial experiment. The payoff under the prior is the pooling payoff, and so \textit{ex ante} learning can only aid the receiver. The second case is trickier: there are two sub-cases that we need to examine: 

\vspace{.2cm}

\textbf{Case 2a:} The mixed strategy of each type has support on at least one message that induces a belief that is \textit{not} on the boundary $R_{1} \cap R_{2}$.

\vspace{.2cm}

\textbf{Case 2b:} There is one type, say $\theta_{L}$, that mixes only over messages that induce beliefs that are on the boundary $R_{1} \cap R_{2}$.

In case 2a, it is easy to see that there must also be a receiver optimal equilibrium in which $\theta_{H}$ and $\theta_{L}$ each choose one message (possibly the same message) that induces a belief in $R \setminus R_{2}$, and $\theta_{M}$ chooses one message that induces a belief in $R \setminus R_{1}$. This is clearly an equilibrium, since each type already has support of its mixed strategy on its respective message; is optimal for the receiver, since this yields the receiver the maximum possible payoff (the separating payoff); and, moreover, does not depend on the prior. Hence, any initial experiment benefits the receiver. 

In case 2b, there must exist a receiver-optimal equilibrium in which $\theta_{H}$ sends just one message, $m_{H}$, which induces a belief in $R \setminus R_{2}$; $\theta_{L}$ sends just one message, $m_{L}$, which induces a belief on the boundary $R_{1} \cap R_{2}$; and $\theta_{M}$ mixes between $m_{L}$ and $m_{M}$, the latter which induces a belief in $R \setminus R_{1}$. We will return to this distribution of posteriors shortly.

Case 3 also must be divided into two cases: 

\vspace{.2cm}

\textbf{Case 3a:} $\theta_{L}$ mixes only over messages that induce beliefs that are on the boundary $R_{1} \cap R_{2}$.

\vspace{.2cm} 

\textbf{Case 3b:} $\theta_{L}$ mixes over at least one message that induces a belief in $R \setminus R_{1}$.

Case 3a is identical to case 2b. In case 3b, there must exist a receiver-optimal equilibrium in which $\theta_{H}$ sends just one message, $m_{H}$, that induces a belief in $R \setminus R_{2}$; and $\theta_{L}$ and $\theta_{M}$ pool on one message $m_{p}$, that induces a belief in $R \setminus R_{1}$. Here, only two messages are used and so by Lemma \ref{too} any initial experiment benefits the receiver.

Consequently, it remains to consider the scenario that case 2b reduces to: for prior $\mu_{0}$ just three messages are used as follows: $\theta_{H}$ separates and chooses message $m_{H}$, $\theta_{L}$ chooses message $m_{L}$ and $\theta_{M}$ mixes between two messages, $m_{L}$ and $m_{M}$, in such a way that the receiver is indifferent over her actions following $m_{L}$ (note, that there is an equivalent scenario that is obtained by interchanging $\theta_{H}$ and $\theta_{L}$).

The prior $\mu_{0}$ must be such that \[\mu_{L}^{0} \leq \frac{w_{M}-v_{M}}{v_{L}-w_{L}}\mu_{M}^{0}\]
and the receiver's payoff is
\[v \coloneqq V(\mu_{0}) = \mu_{H}^{0}v_{H} + \mu_{M}^{0}w_{M} + \mu_{L}^{0}w_{L}\]
Call this equilibrium \hypertarget{b}{\textcolor{Fuchsia}{$S^{\dagger}$}}.
It is clear that without loss of generality we may focus on an initial experiment that is binary, and which yields just two beliefs $\mu_{1}$ and $\mu_{2}$, where $\mu_{1}$ is a belief such \hyperlink{b}{$S^{\dagger}$} is feasible, and $\mu_{2}$ is a belief such that $S^{\dagger}$ is infeasible. To see that this is without loss of generality, note that if there are multiple initial experiment realizations after which $S^{\dagger}$ is feasible, the receiver achieves at least the payoff as in the case when there is just one such initial experiment realization. Likewise, if there are multiple initial experiment realizations after which $S^{\dagger}$ is infeasible, since we need only assume the pooling payoff in this case, it is again clear that the receiver achieves at least the payoff as in the case where there is just one such initial experiment realization.

Thus, the initial experiment, $\zeta$, yields $\mu_{1} = \left(\mu_{L}^{1}, \mu_{M}^{1}, \mu_{H}^{1}\right)$ with probability $p$ and $\mu_{2} = \left(\mu_{L}^{2}, \mu_{M}^{2}, \mu_{H}^{2}\right)$ with probability $(1-p)$, where $p \mu_{1} + (1-p) \mu_{2} = \mu_{0}$. For belief $\mu_{2}$, the receiver's payoff is bounded below by the pooling equilibrium payoff, and so we assume that that is indeed the payoff. Note that since $\mu_{2}$ is not a belief for which $S^{\dagger}$ is feasible, we must have \[\tag{$A11$}\label{18}\mu_{L}^{2} > \frac{w_{M}-v_{M}}{v_{L}-w_{L}}\mu_{M}^{2}\]
\begin{claim}
For belief $\mu_{2}$, action $a_{1}$ is optimal.
\end{claim}
\begin{proof}
Suppose for the sake of contradiction that $a_{1}$ is not optimal. That is
\[\mu_{L}^{2}w_{L} + \mu_{M}^{2}w_{M} + \mu_{H}^{2}w_{H} > \mu_{L}^{2}v_{L} + \mu_{M}^{2}v_{M} + \mu_{H}^{2}v_{H}\]
But then 
\[\begin{split}
    \mu_{L}^{2}w_{L} + \mu_{M}^{2}w_{M} + \mu_{H}^{2}w_{H} &> \mu_{L}^{2}w_{L} + \mu_{M}^{2}w_{M} + \mu_{H}^{2}v_{H}\\
    \mu_{H}^{2}w_{H} &> \mu_{H}^{2}v_{H}\\
\end{split}\]
where the second inequality follows from Inequality \ref{18}. This is a contradiction.
\end{proof}
Thus, $a_{1}$ is optimal and so the receiver's expected payoff is
\[V = p\left[\mu_{H}^{1}v_{H} + \mu_{M}^{1}w_{M} + \mu_{L}^{1}w_{L}\right] + (1-p)\left[\mu_{L}^{2}v_{L} + \mu_{M}^{2}v_{M} + \mu_{H}^{2}v_{H}\right]\]
which reduces to
\[V = v + (1-p)\left[\mu_{L}^{2}\left(v_{L}-w_{L}\right) - \mu_{M}^{2}\left(w_{M}-v_{M}\right) \right]\]
which is greater than $v$ by Inequality \ref{18}. Case 2b is illustrated in Figure \ref{3ow}.

\end{proof}

\subsection{Lemma \ref{410} Proof and Payoff Function Derivation}\label{410proof}
\begin{proof} We derive the receiver's payoff as a function of the belief $\mu$ through a pair of claims. First,
\begin{claim}
For any belief $\mu > 13/36$, there exists no equilibrium in which a message is played that induces a belief such that action $a_{1}$ is strictly optimal.
\end{claim}
\begin{proof}

We can exhaustively proceed through each message:

\begin{enumerate}
    \item Suppose $a_{1}$ is strictly optimal following $m_{1}$. Then, $\theta_{4}$ must have support of his mixed strategy on $m_{1}$. That gives him a payoff of $-1$, so he can deviate profitably to $m_{2}$. 
    \item Suppose $a_{1}$ is strictly optimal following $m_{2}$. Both $\theta_{3}$ and $\theta_{4}$ must have support of their mixed strategies on $m_{2}$.  Moreover, so much of their support must be on $m_{2}$ that the receiver must choose $a_{2}$ following $m_{1}$ (which will always be chosen by $\theta_{1}$). Hence, $\theta_{3}$ can deviate profitably to $m_{1}$.
    \item Suppose $a_{1}$ is strictly optimal following $m_{3}$. Consequently $\theta_{3}$ must have some support of his mixed strategy on $m_{3}$ (since $\theta_{4}$ will never choose $m_{3}$). Moreover, $\theta_{3}$ cannot be choosing a pure strategy since otherwise he would have a profitable deviation to $m_{2}$. Hence, $\theta_{3}$ must be mixing over $m_{3}$ and $m_{2}$ (since $3$ is strictly larger than either of $\theta_{3}$'s payoffs for $m_{1}$). The receiver must also mix following $m_{2}$, so as to leave $\theta_{3}$ willing to mix. In particular, the receiver must choose $a_{2}$ with probability $2/3$ following $m_{2}$ and $a_{1}$ with probability $1/3$. But note that $\theta_{4}$ must also have support on $m_{2}$, which message would thus yield it a payoff of $2/3$, which is less than $5/4$, the payoff he would get from deviating profitably to $m_{1}$.
\end{enumerate}
\end{proof}
As a result $V^{T} = \frac{1}{3} + \mu$ for all $\mu > 13/36$. Second,
\begin{claim}
For any belief $\mu \leq 13/36$, the receiver optimal equilibrium begets a payoff of $37/24 - 2\mu$.
\end{claim}
\begin{proof}
First, an equilibrium that begets such a payoff exists. Type $\theta_{1}$ chooses $m_{1}$, types $\theta_{2}$ and $\theta_{4}$ choose $m_{2}$, and type $\theta_{3}$ chooses $m_{3}$. Upon observing $m_{1}$, the receiver chooses $a_{2}$, and upon observing either $m_{2}$ or $m_{3}$ the receiver chooses $a_{1}$.

It is immediately evident that neither $\theta_{1}$ nor $\theta_{2}$ have profitable deviations, since they are choosing strictly dominant strategies. On path, $\theta_{4}$ obtains $2$, whereas his payoff from deviating would be less than $2$. Finally, $\theta_{3}$ obtains $3$ following $m_{3}$ and less than $3$ following any other message. The receiver's equilibrium payoff is
\[\frac{1}{3} + \frac{1}{8} + 2\left(\frac{13}{24}-\mu\right) = \frac{37}{24} - 2\mu\]

Second, it is easy to verify that this equilibrium is optimal for the receiver: $\theta_{4}$ is unwilling to choose message $m_{3}$ and $\theta_{1}$ and $\theta_{2}$ always choose $m_{1}$ and $m_{2}$, respectively. Thus, the receiver will always get her decision ``wrong" with respect to some mixture of $\theta_{1}$, $\theta_{2}$ or $\theta_{4}$. For $\mu \leq 1/3$, she prefers to get her decision wrong with respect to $\theta_{2}$, so our proposed equilibrium is obviously optimal. For $\mu \in (1/3, 13/36]$, she prefers to get her decision wrong with respect to $\theta_{1}$ but then $\theta_{4}$ would always have a profitable deviation to $m_{2}$. Thus, she gets her decision wrong with respect to $\theta_{2}$.

Alternatively, as noted in Whitmeyer (2019) \cite{Whit}, this is the receiver-optimal payoff under any information structure (or any degree of transparency). Therefore, since the receiver can obtain this payoff with full transparency, the corresponding equilibrium must be optimal.
\end{proof}
Finally, Figure \ref{4dconvex} illustrates that $V^{T}$ is not convex.
\end{proof}

\subsection{Lemma \ref{3by3} Proof}\label{3by3proof}

\begin{proof}

It is easy to see that in this game there are no separating equilibrium--type $\theta_{L}$ can deviate profitably by mimicking either of the other types. For each of the three possible priors ($\mu_{0}$, $\mu_{1}$, and $\mu_{2}$) any pooling equilibrium results in the receiver choosing $l$. There exist pooling equilibria for each prior, and any off-path belief sustains such equilibria, since the sender types receive their maximal payoff on path.

Compiling the remaining equilibria appears daunting (or at least unpleasant), but we can thankfully bypass this, since we are interested in the receiver optimal equilibria. Instead of finding each equilibrium, we take a belief-based approach and construct and solve the appropriate maximization problem for the receiver.

Observe that the receiver has at most three posterior distributions at equilibrium, which corresponds to the case in which all three messages are used on path, which messages beget different posterior distributions. Note that any distribution or belief in the context of this example corresponds to a point in Figure \ref{convex1}. It is a standard result that beliefs are a \textit{martingale}--hence, an equilibrium pair of posterior beliefs is feasible only if there exists a line segment between the two beliefs that intersects the prior (for it to be an \textit{equilibrium} pair of beliefs, each posterior must be generated by an equilibrium vector of strategies). Likewise, a triplet of posteriors is feasible only if the prior lies in the convex hull of the three points. 

From this, we see immediately that for each (prior) belief, $\mu_{0}$, $\mu_{1}$, and $\mu_{2}$, it is impossible for none of the equilibrium posterior beliefs to lie in $l$. Even more, at least one of the posterior beliefs cannot lie on the boundary between $l$ and $s$ or the boundary between $l$ and $r$. Hence, there are four cases: i. All of the posteriors lie in $l$, ii. Some of the posteriors lie in $s$ and some in $l$, iii. Some of the posteriors lie in $x$ and some in $l$, or iv. There is one posterior in each of $l$, $s$ and $x$.

However, cases iii and iv are impossible: each type strictly prefers $l$ or $s$ to $x$ and thus some type must have a profitable deviation to a message that induces a posterior in $l$. Moreover, if all of the posteriors lie in $l$, then this yields the same payoff to the receiver as the case in which each type pools and so we may ignore case i.

Consequently, it remains to consider case ii. In addition, note that without loss we may focus on the situation in which there are just two posteriors since if the optimum consisted of three posteriors, two of them must lie in the same region, and we could just take their average, which would also lie in the same region and yield the receiver the same (total) payoff.

As a result, for each (prior) belief $Q$, $Q \in \left\{\mu_{0}, \mu_{1}, \mu_{2}\right\}$, the receiver solves 
\[\max_{\lambda, \mu_{1},\mu_{2}}\left\{\lambda V(\mu_{1}) + (1-\lambda)V(\mu_{2})\right\}\]
subject to \[\lambda \mu_{1} + (1-\lambda) \mu_{2} = Q\]
and $\mu_{1} \in B$, $\mu_{2} \in A$, and $\lambda \in [0,1]$. Substituting in the payoffs, and with the aid of Figure \ref{convex1}, we have
\[\max_{\lambda, \mu_{H}^{1}, \mu_{M}^{1},\mu_{H}^{2}, \mu_{M}^{2}}\left\{\lambda \left(\mu_{H}^1 + \frac{13}{24}(1-\mu_{H}^1)\right) + (1-\lambda)(2\mu_{H}^{2} + \mu_{M}^2)\right\}\]
subject to \[\begin{split}
    \lambda \mu_{H}^{1} + (1-\lambda) \mu_{H}^{2} &= \mu_{H}^{0}, \qquad
    \lambda \mu_{M}^{1} + (1-\lambda) \mu_{M}^{2} = \mu_{M}^{0},\\
    1 &\geq \lambda \geq 0, \qquad
    \mu_{H}^{1} \geq 0,\\
    \frac{13-37\mu_{H}^{1}}{24} &\geq \mu_{M}^{1} \geq \frac{11}{24}(1-\mu_{H}^{1}),\\
    \mu_{M}^{2} &\geq 0, \qquad
    \mu_{H}^{2} + \mu_{H}^{2} \leq 1\\
\end{split}\]
For belief $\mu_{0}$, $\lambda^{*} = 6/13$, and the optimal pair of equilibrium posteriors is \[\label{A111}\tag{$A12$}\left(\frac{13}{24}, \frac{11}{24}, 0\right), \quad \text{and} \quad \left(0, \frac{1}{14}, \frac{13}{14}\right)\] which yields the receiver a payoff of $67/52$. This corresponds to an equilibrium in which $\theta_{H}$ and $\theta_{L}$ choose different messages, say $g$ and $b$, respectively; and $\theta_{M}$ mixes between those messages ($g$ and $b$). 

For belief $\mu_{1}$, $\lambda_{1}^{*} = 2/13$, and the optimal pair of equilibrium posteriors is \[\label{A112}\tag{$A12$}\left(\frac{13}{24}, \frac{11}{24}, 0\right), \quad \text{and} \quad \left(0, \frac{7}{33}, \frac{26}{33}\right)\] which yields the receiver a payoff of $83/52$. This corresponds to the same type of equilibrium as for $\mu_{0}$ (the receiver can distinguish between $\theta_{H}$ and $\theta_{L})$.

For belief $\mu_{2}$, $\lambda_{2}^{*} = 6/11$, and the optimal pair of equilibrium posteriors is \[\label{A113}\tag{$A13$}\left(\frac{13}{24}, \frac{11}{24}, 0\right), \quad \text{and} \quad \left(\frac{4}{15}, 0, \frac{11}{15}\right)\] which yields the receiver a payoff of $127/132$. This corresponds to an equilibrium in which $\theta_{H}$ and $\theta_{M}$ choose different messages, say $g$ and $m$, respectively; and $\theta_{L}$ mixes between those messages ($g$ and $m$).

The pooling equilibrium payoffs for each prior are $5/4$, $19/12$, and $11/12$, for $\mu_{0}$, $\mu_{1}$, and $\mu_{2}$, respectively. Hence, the pooling equilibrium is not optimal for any of the priors, and so the equilibria that correspond to the posterior pairs given in Expressions \ref{A111}, \ref{A112}, and \ref{A113} are optimal for their respective priors. In Lemma 4.4 in Whitmeyer (2019) \cite{Whit}, we determine that these are the maximal payoffs for the receiver at these beliefs under \textit{any} information structure, which corroborates the optimality of these equilibria.

It remains to verify
\[\frac{67}{52} > \frac{2195}{1716} = \frac{1}{2}\cdot \frac{83}{52} + \frac{1}{2}\cdot \frac{127}{132}\]
and thus the receiver's optimal equilibrium payoff is not convex in the prior. Note that this game is a cheap talk game with transparent motives--even these restrictions are not enough to guarantee convexity. 
\end{proof}

\end{document}